\documentclass[review]{elsarticle}
\usepackage{multirow}
\usepackage{url}
\usepackage{algorithm}
\usepackage[noend]{algpseudocode}
\usepackage{subfig}
\usepackage{soul}
\usepackage[colorinlistoftodos,prependcaption,textsize=tiny]{todonotes}
\usepackage{amsmath}
\usepackage{amsthm}
\usepackage{dcolumn}
\usepackage{graphicx}
\usepackage{algorithmicx}
\usepackage{algpseudocode}   
\usepackage{booktabs,siunitx}

\usepackage{topcapt}
\usepackage{booktabs}
\usepackage{flushend}

\usepackage{caption}
\usepackage{setspace}

\PassOptionsToPackage{round, sort,numbers}{natbib} 

\setcitestyle{authoryear,open={(},close={)}}

\journal{Expert Systems With Applications}

\begin{document}
\newtheorem{corollary}{Corollary}
\newtheorem{problem}{Problem}
\newtheorem{observation}{Observation}
\newtheorem{example}{Example}
\newtheorem{lemma}{Lemma}
\newtheorem{theorem}{Theorem}
\newtheorem{defin}{Definition}

\newcommand{\single}{{\sc SingleTeam}}
\newcommand{\partition}{{\tt Faultline-Partitioning}}
\newcommand{\reassignteams}{\ensuremath{\texttt{ReassignTeams}}}
\newcommand{\assigncosts}{\ensuremath{\texttt{AssignCosts}}}
\newcommand{\iterative}{{\tt Iterative}}
\newcommand{\plusiterative}{{\tt Iterative+}}
\newcommand{\minusiterative}{{\tt Iterative-}}
\newcommand{\homo}{{\tt Homogeneous}}
\newcommand{\hetro}{{\tt Heterogeneous}}
\newcommand{\hybrid}{{\tt Hybrid}}
\newcommand{\networkx}{{NetworkX}}
\newcommand{\freelancer}{{\sl Freelancer}}
\newcommand{\bc}{{\ensuremath{\mathbf{c}}}}
\newcommand{\bT}{{\ensuremath{\mathbf{T}}}}
\newcommand{\aggf}{{\ensuremath{r}}}
\newcommand{\worker}{{\ensuremath{\mathbf{w}}}}
\newcommand{\workers}{{\ensuremath{W}}}
\newcommand{\sign}{{\ensuremath{\texttt{sign}}}}
\newcommand{\calW}{{\ensuremath{\mathcal{W}}}}
\newcommand{\calR}{{\ensuremath{\mathcal{R}}}}
\newcommand{\calQ}{{\ensuremath{\mathcal{Q}}}}
\newcommand{\calP}{{\ensuremath{\mathcal{P}}}}
\newcommand{\calN}{{\ensuremath{\mathcal{N}}}}
\newcommand{\calF}{CT}
\newcommand{\calA}{{\ensuremath{\mathcal{A}}}}
\newcommand{\matching}{\ensuremath{\texttt{FaultlineSplitter}}}
\newcommand{\genmatching}{\ensuremath{\texttt{MemberExchange}}}
\newcommand{\grd}{\ensuremath{\texttt{Greedy}}}
\newcommand{\nbt}{\ensuremath{\texttt{TriadicM-F}}}
\newcommand{\npt}{\ensuremath{\texttt{TriadicM-P}}}
\newcommand{\nnt}{\ensuremath{\texttt{TriadicM-N}}}
\newcommand{\dya}{{\tt Clustering}}
\newcommand{\greedy}{{\tt Greedy}}

\newcommand{\synthetic}{{\sf Synthetic-1}}
\newcommand{\synthetictwo}{{\sf Synthetic-2}}
\newcommand{\bia}{{\sf BIA660}}
\newcommand{\dblpAug}{{\sf DBLP-Aug}}
\newcommand{\adult}{{\sf Adult}}
\newcommand{\census}{{\sf Census}}
\newcommand{\dblp}{{\sf DBLP}}
\newcommand{\ct}{{\sf CT}}
\newcommand{\asw}{{\sf ASW}}
\newcommand{\sss}{{\sf SS}}

\newcommand{\etal}{{\it et. al.}}

\newcommand{\spara}[1]{\smallskip\noindent{\bf{#1}}}
\newcommand{\mpara}[1]{\medskip\noindent{\bf{#1}}}
\newcommand{\bpara}[1]{\bigskip\noindent{\bf{#1}}}

\newcommand{\squishlist}{\begin{list}{$\bullet$}
  { \setlength{\itemsep}{0pt}
     \setlength{\parsep}{3pt}
     \setlength{\topsep}{3pt}
     \setlength{\partopsep}{0pt}
     \setlength{\leftmargin}{1.5em}
     \setlength{\labelwidth}{1em}
     \setlength{\labelsep}{0.5em} } }
\newcommand{\squishend}{
  \end{list}  }

\begin{frontmatter}

\title{A Team-Formation Algorithm \\for Faultline Minimization}

\author[bu]{Sanaz Bahargam}
\ead{bahargam@cs.bu.edu}
\author[bu]{Behzad Golshan}
\ead{behzad@cs.bu.edu}
\author[st]{Theodoros Lappas}
\ead{tlappas@stevens.edu}
\author[bu]{Evimaria Terzi}
\ead{evimaria@cs.bu.edu}

\address[bu]{Boston University}
\address[st]{Stevens Institute of Technology}

\begin{abstract}
In recent years, the proliferation of online resumes and the need to evaluate large populations of candidates 
for on-site and virtual teams have led to a growing interest in automated team-formation. Given a large pool of candidates,
the general problem requires the selection of a team of experts to complete a given task. Surprisingly, while ongoing research has studied numerous variations with different constraints, it has overlooked a factor with a well-documented
impact on team cohesion and performance: team faultlines. Addressing this gap is challenging, as the available measures for faultlines in existing 
teams cannot be efficiently applied to faultline optimization. In this work, we meet this challenge with a new measure
that can be efficiently used for both faultline measurement and minimization. We then use the measure to solve the problem of automatically partitioning a large population into low-faultline teams.
By introducing faultlines to the team-formation literature, our work creates exciting opportunities for algorithmic work on faultline
optimization, as well as on work that combines and studies the connection of faultlines with other influential team characteristics. 
\end{abstract}

\begin{keyword}
Teams\sep Team Faultlines \sep Team Formation
\end{keyword}

\end{frontmatter}


\section{Introduction}
The problem of organizing the individuals in a given population into teams emerges in 
multiple domains. 
In a business setting, 
the workforce of a firm is organized in groups, with each group dedicated to a different 
project~\citep{mohrman1995designing}. In an educational context,
it is common for the instructor to partition the students in her class into small teams, with 
team members collaborating to complete 
different types of assignments~\citep{agrawal14grouping,webb1982student,bahargam2017team,agrawal2014forming}. In a 
government setting, elected officials are organized in committees that 
design and implement policies for a wide spectrum of critical 
issues~\citep{fenno1973congressmen}. 

In recent years, the proliferation of online resumes and the need to evaluate large populations of candidates 
for on-site and virtual teams have led to a growing interest in automated team-formation~\citep{lappas09finding,anagnostopoulos12online,golshan14profit-maximizing,majumder12capacitated,kargar11teamexp,an13finding,dorn10composing,gajewar12multiskill,li10team,sozio10community,agrawal14grouping,bahargampersonalized}. Given a large pool of candidates,
the general problem requires the selection of a team of experts to complete a given task. 
The ongoing literature has studied numerous problem variations with different constraints and optimization criteria.  
Examples include the  coverage of all the skills required to achieve a set of goals~\citep{lappas09finding,li10team,gajewar12multiskill}, smooth communication among the members of the team~\citep{rangapuram2013towards,anagnostopoulos12online,lappas09finding,kargar2012efficient}, the minimization of the cost of recruiting promising candidates~\citep{golshan14profit-maximizing,kargar2012efficient},
scheduling constraints~\citep{durfee2014using}, the balancing of the workload assigned to each member~\citep{majumder12capacitated}, and the need for effective  leadership~\citep{kargar11teamexp}. 

Surprisingly, while ongoing research on team formation has studied numerous variations with different constraints, it has overlooked a factor with a well-documented impact on a team cohesion and performance: \emph{team faultlines}. 
The faultline concept was introduced in the seminal work by Lau and Murnighan~\citep{lau1998demographic}. 
Faultlines manifest as hypothetical dividing lines that split a group into relatively homogeneous subgroups based on multiple attributes~\citep{lau1998demographic,meyer2013team}. The consideration of multiple attributes is critical, as it distinguishes relevant work from the study on  single-attribute faultlines, referred to as a ``separation''~\citep{harrison2007s}. The team-formation framework that we describe in this work focused on the general multi-attribute paradigm.
%
%
Faultline-caused subgroups are in risk of colliding, leading to costly conflicts, poor communication, and disintegration~\citep{bezrukova2009workgroup,choi2010group,gratton07bridging,jehn2010faultline,li2005factional,molleman2005diversity,polzer2006extending,shaw2004development,
thatcher2003cracks}.

Bridging the faultline literature with automated team-formation is challenging, as the available measures for faultlines in existing 
teams cannot be efficiently used for faultline optimization. For instance, many faultline measures utilize clustering algorithms
to identify the large homogenous groups that create faultlines within a team~\citep{meyer2014faultlines,jehn2010faultline,meyer2013team,barkema2007does,lawrence2011identifying}. While such measures have emerged as the state-of-the-art, their clustering step requires a pre-existing team. Therefore, in a team-formation setting, a clustering-based measure 
would need to naively consider all (or an exponential number) of possible teams in order to find a faultline-minimizing solution. 
This brute-force approach is not applicable to even moderately-sized populations. Similarly, we cannot use any of the existing measures based on expensive (and often exponential) computations to identify the subgroups within a team~\citep{thatcher2003cracks,zanutto2011revisiting,bezrukova2009workgroup,trezzini2008probing,shaw2004development,van2011diversity}.

In this work, we describe the fundamental efficiency principles that a faultline measure needs to follow in order to be 
applicable to the automated formation of faultline-minimizing teams. We then introduce  \emph{Conflict Triangles} ({\ct}): a new measure that follows these principles. The {\ct} measure is consistent with the principles of faultline theory~\citep{lau1998demographic} and is founded on the extensive work on the balance of social structures~\citep{cartwright56structural,easley10networks,heider58psychology,morrissette67no}.
We then use this measure as the objective function for the problem of partitioning a given population into teams, such that the average faultline score per team is minimized. We refer to this as the {\partition} problem and formally define it in Section~\ref{sec:problem}. 
Our work 
thus makes the following contributions:
\begin{enumerate}
\item{We initiate research on the unexplored overlap between the decades of work on team faultlines and the rapidly emerging field
of automated team formation.}
\item{We describe the fundamental efficiency principles that a faultline measure has to satisfy to be applicable to faultline-aware team-formation.}
\item{We present a new measure that follows these principles can thus be used for both faultline measurement and minimization. Our evaluation demonstrates the measure's effectiveness in both tasks.}
\item{We formally define the {\partition} problem, analyze its complexity, and present an efficient algorithmic framework for its solution.}
\end{enumerate}

By introducing faultlines to automated team-formation, our work creates exciting opportunities for algorithmic work on faultline
optimization, as well as on work that combines and studies the connection of faultlines with other influential team characteristics. 
In Section~\ref{sec:conclusions}, we discuss the implications of our work for practitioners in both organizational and educational settings and discuss potential directions for future work.

\section{Background and Motivation}\label{sec:related}
To the best of our knowledge, our work is the first to incorporate faultlines in an algorithmic framework
for automated team-formation. However, our work is related to three types of research:
$(i)$ algorithmic frameworks for optimizing 
various factors that affect the performance of a team. 
$(ii)$ management, psychology and sociology studies
on faultlines and their effects on team outcomes, and $(iii)$  efforts measure faultlines in existing teams. Next, we discuss each of these categories in more detail.

\subsection{Algorithmic work on team formation}
\label{sec:auto}
Our previous work~\citep{lappas09finding} studied the problem of automated team-formation in the context of social networks.
Given a pool of experts and a set of skills that needed to be covered, the goal there is to select a team of experts that can collectively cover all the required skills, while ensuring efficient intra-team
communication.
Over the last years, this work has been extended to 
identify a single team or a collection of teams that optimize  different
factors that influence a team's performance. 
For example, a significant body of work has focused on incorporating different 
definitions of the communication cost among experts~\citep{anagnostopoulos12online, an13finding,dorn10composing,gajewar12multiskill,kargar11teamexp,li10team,sozio10community,galbrun2017finding}.  
Other work has also focused on optimizing the cost of recruiting promising candidates~\citep{golshan14profit-maximizing,an13finding},  minimizing the workload assigned to 
each individual team member~\citep{majumder12capacitated,anagnostopoulos10power}, satisfying scheduling constraints~\citep{durfee2014using}, identifying effective leaders~\citep{kargar11teamexp}, and optimizing the individual's benefit from team participation~\citep{agrawal14grouping,bahargampersonalized}.
Although all these efforts focus on optimizing various teams aspects, the work that we describe in this paper is the first
to address faultline optimization. As we describe in our work, minimizing faultline potential raises new algorithmic challenges that cannot be addressed by extant algorithmic solutions.

\subsection{Studies on the effects of team faultlines}
For decades, researchers from various disciplines have studied the creation, operation, and performance of teams in different settings. 
Faultline theory was introduced by Lau and Murnighan~\citep{lau1998demographic}. It has since been the focus of numerous follow-up works. A number of papers have studied 
how the existence of faultlines within a team can lead to conflict~\citep{li2005factional,choi2010group,thatcher2003cracks} and affect functionality~\citep{molleman2005diversity,polzer2006extending} and performance~\citep{bezrukova2009workgroup,thatcher2003cracks}. Motivated by the observation that the existence of faultlines does not guarantee the formation of colliding subgroups, researchers have also studied the factors that 
can lead to faultline activation~\citep{pearsall2008unlocking,jehn2010faultline}. Further, Gratton et al.~\citep{gratton07bridging} explored strategies that a leader or manager can follow to effectively handle or avoid the emergence of faultlines within a team.

\subsection{Operationalizing Faultline Strength}
\label{sec:operationalize}
Previous work has suggested various methods for evaluating faultlines in teams. 
Even though the original faultline paper by ~\citep{lau1998demographic} serves as the foundation of the long line of relevant literature and introduces principles that we also adopt in our work, it does not define a faultline measure. Instead, the authors 
lay out fundamental principles that a measure needs to follow in order to accurately evaluate the faultline strength 
in a given team. While these principles are appropriate for faultline measurement, they are not sufficient to ensure that a qualifying
measure will also have the computational efficiency required to serve as the objective function of a scalable  algorithm
that has to process large populations of candidates to create teams with minimal faultlines. Computational efficiency is critical
in this setting, as each of the hundreds or thousands of individuals in the given population can be represented by a point in a multidimensional space of attributes (e.g. demographics, resume information). Any team-formation algorithm would then have to efficiently navigate this space and quickly evaluate the faultline strength of many different combinations in order to identify faultline-minimizing teams. Therefore, in order to be efficiently applicable to faultline minimization, a faultline measure should follow the following two efficiency principles:
\begin{itemize}
\item{\textbf{Linear Computation:} The measure should be easy to compute for a given team in polynomial time.}
\item{\textbf{Constant Updates:} The measure should be easy to update in constant time if one person joins or leaves the team.}
\end{itemize}

In Section~\ref{sec:prelim}, we introduce  \emph{Conflict Triangles} ({\ct}): a new measure that provides these two characteristics. The {\ct} measure is consistent with the principles of faultline theory~\citep{lau1998demographic} and is founded on the extensive work on the balance of social structures~\citep{cartwright56structural,easley10networks,heider58psychology,morrissette67no}.
Next, we review the extensive literature on faultline measurement and discuss the shortcomings of extant measures in the context of the two efficiency principles that are necessary for automated team-formation.

\subsubsection{State of the Art in Faultline measurement}
A long line of literature has focused on identifying and measuring the strength of faultlines in existing teams. In recent years,
clustering-based algorithms have emerged as the state of the art for this purpose~\citep{meyer2014faultlines,jehn2010faultline,meyer2013team,barkema2007does,lawrence2011identifying}. This line of work is exemplified by the 3-step Average Silhouette Width ({\asw}) approach proposed by~\citep{meyer2013team}. 
Given a team of individuals, the first step includes applying an agglomerative-clustering algorithm for pre-clustering the team's members. Agglomerative clustering begins by assigning each member to its own cluster. The two most similar clusters are then iteratively joined until all points belong to the same cluster. 
The authors of the original paper experiment with the two
most popular merging criteria: Ward's algorithm and Average Linkage (AL). Thus, for a team with $n$ members, the joint set of results from the two alternatives yields a total of $2\times n$ possible configurations ($2$ for each possible number of clusters).

The second step focuses on computing the {\asw} of each possible configuration~\citep{rousseeuw1987silhouettes}. The silhouette $s(i)$ of an individual $i$ quantifies how well a team member $i$ fits into its cluster in comparison to all other clusters and is formally defined as:
$$
s(i)=\frac{b_i-\alpha_i}{max(\alpha_i,b_i)},
$$
where $a_i$ is the average distance of $i$ to all other point in its cluster and $b_i$ is the lowest average distance of $i$ to all points in any other cluster of which $i$ is not a member.  The silhouette ranges from $-1$ to $+1$, where a high value indicates that the object is well matched to its own cluster and poorly matched to others. The {\asw} is the average silhouette of all the team's members.

The third step employs a post-processing method to maximize the {\asw}  of each configuration, by temporarily moving individuals across subgroups and recomputing the {\asw} after each move.
The move that leads to the highest increase is made permanent. The process continues until no further improvement is possible. Finally,
the maximum {\asw} score over all configurations is reported as the strength of the team's faultline structure.

\spara{Using {\asw} in automated team-formation:} Previous work has repeatedly verified the advantage of the {\asw} measure over alternative approaches~\citep{meyer2014faultlines,meyer2013team}. However, the measure cannot be efficiently used as the objective function for team-formation algorithms, as it is designed to evaluate faultline strength in existing teams and assumes that the composition of a team is part of the input. 
As stated earlier in this section, an appropriate measure for faultline minimization should be easy to compute in linear
time and easy to update in constant time. However, Given a team of $n$ individuals,
the complexity of the agglomerative-clustering step alone is $O(n^2logn)$~\citep{rokach2005clustering}. There is then no guarantee on the number of reassignments that it will take for the {\asw} score to converge. In addition, the score cannot be updated in constant time.
Instead, the deletion or addition of a member would require the new team to be re-evaluated from scratch, in order to compute the optimal {\asw} score.

In theory, a practitioner could consider all possible teams, evaluate their respective faultline strengths, and choose the optimum. In practice, however, this brute-force approach is not scalable and can only be applied to populations of trivial size. It can certainly not be applied to populations of hundreds or thousands of individuals, which are common in the team-formation literature~\citep{lappas09finding,anagnostopoulos12online,anagnostopoulos10power}. The team-partitioning task that we address in this work is considerably
more computationally challenging than single-team formation. In order to use the {\asw} measure for this task, a practitioner would
have to consider all possible partitionings of a population into fixed-size, non-overlapping teams. This is a computationally
intractable process that would have to consider $O(N!)$ alternatives. 
Similar to the {\asw} measure, other clustering-based approaches are also excluded from automated team-formation due to computational
efficiency~\citep{barkema2007does,lawrence2011identifying}. 

\subsubsection{Other Faultline Measures}
Similar to clustering-based approaches, most existing faultline measures are not applicable to team-formation tasks due to computational efficiency.
For instance, the cost to compute the Index of Polarized Multi-Dimensional Diversity proposed by~\citep{trezzini2008probing} 
grows exponentially with the number of attributes. The SGA measure by~\citep{carton2013impact} depends on the exhaustive evaluation of every possible partition of a given group with two or more subgroups. Similarly, the FLS measure by~\citep{shaw2004development}
depends on the computation and averaging of all possible internal alignments and cross-product alignments of every feature with respect to the subgroups of every other feature. Given that each of these constructs has to be updated every time a person is added to or removed from a team, the FLS formula cannot be updated in constant time. 
The measure proposed by~\citep{van2011diversity} uses regression analysis
to measure the variance of each attribute that is explained by all other attributes. Despite its advantages in a measurement
setting, running multiple regressions for every candidate team is not a realistic option in a team-formation setting.

\citep{thatcher2003cracks} propose a formula for computing the portion of the total variance explained by a given
segmentation of a team into subgroups. Their final faultline measure $Fau_g$ is then defined as the score of the 
segmentation that maximizes the formula.  However, the measure can only be applied for segmentations of two subgroups
due to (i) the exhaustive nature of the search for the best split that makes the cost prohibitive in a team-formation setting, and (ii) the fact that, if we allow the number of subgroups to vary
arbitrarily, the solution that maximizes the formula is to trivially assign each individual to its own subgroup. Hence, an algorithm
that uses this measure to create low-faultline teams would never choose to create highly diverse teams, despite the fact that high diversity is associated with low faultlines~\citep{lau1998demographic}.
These limitations are inherited by follow up efforts that extend this measure~\citep{zanutto2011revisiting,bezrukova2009workgroup}.
The measure by~\citep{li2005factional} assumes a specific attribute of interest and is not suitable for evaluating team faultlines across attributes. This is also the reason that the measure has been excluded by comparative studies of faultline measures~\citep{meyer2013team}.

Another relevant construct is the Subgroup Strength measure proposed by~\citep{gibson2003healthy}. While this measure
is not designed for faultline measurement, it is relevant due to its focus on subgroups.
Its creators posit that strong subgroups exist if there is high
variability in the extent to which attributes overlap in the dyads within a team. Their measure is thus based on computing
the pairwise similarities between the team's members across all attributes. The team's subgroup strength
is then computed as the standard deviation over all possible member pairs. Even though this measure is not specifically designed for faultline measurement, it is easy to compute and to update, as required by the team-formation paradigm. Hence, we include this measure in our experimental evaluation in Section~\ref{sec:experiments}.

\section{Operationalizing a Team's Faultline Potential}
\label{sec:prelim}

We consider a pool {\workers} of $n$ individual workers.
Each worker $i\in \workers$ is associated with an $m$-dimensional
feature vector $\worker_i$, such that $\worker_i(f)$ returns the value
of feature $f$ for worker $i$. For each feature $f$, we create a complete \emph{signed} graph $G_f$
that includes one node for each worker in $\workers$. The sign of the edge between two nodes (workers) $(i,i')$ is positive if 
they have the same value for feature $f$ (i.e. $\worker_i(f)=\worker_{i'}(f)$ ) and negative otherwise. 
Consider the following example:

\begin{example}
\label{ex:featuregraph}
We are given a pool of $3$ workers, where each worker is described by $3$
features: \emph{country of origin, gender}, and \emph{undergraduate major}. Our data thus consists
of the following feature vectors:
\begin{eqnarray*}
\worker_1 &=& [\mathrm{\tt{India}}, \mathrm{\tt{Male}},
\mathrm{\tt{Computer\ Science}}] \\
\worker_2 &=& [\mathrm{\tt{India}}, \mathrm{\tt{Male}},
\mathrm{\tt{Business}}] \\
\worker_3 &=& [\mathrm{\tt{China}}, \mathrm{\tt{Male}},
\mathrm{\tt{Chemistry}}]
\end{eqnarray*}

Fig~\ref{fig:theseauthors} shows the graphs for the three features.

\begin{figure}
\begin{center}
\begin{tabular}{c}
\subfloat[]{
\includegraphics[scale=.2]{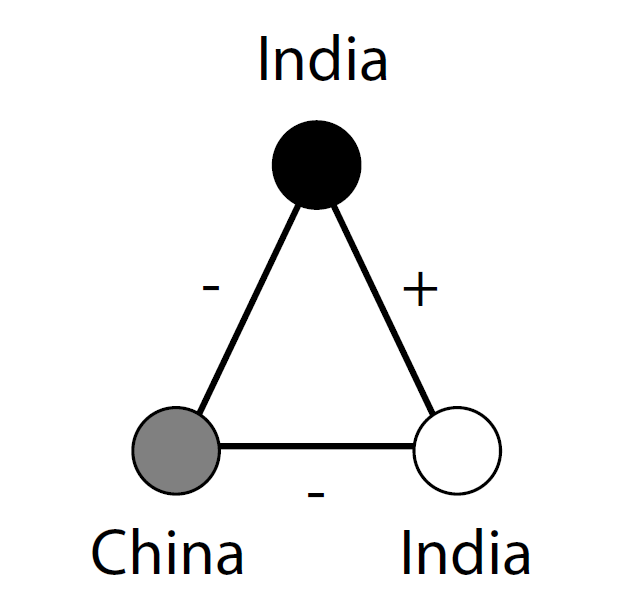}
\label{fig:theseauthors-bad}
} 
\subfloat[]{
\includegraphics[scale=.2]{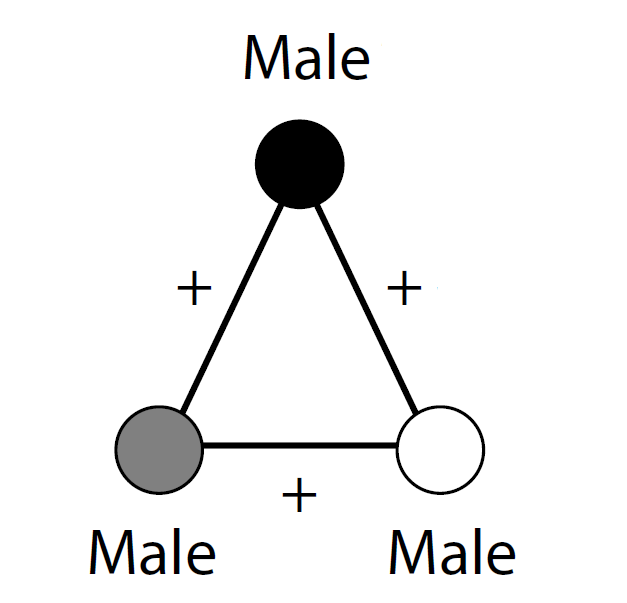}
\label{fig:theseauthors-pos}
}
\subfloat[]{
\includegraphics[scale=.2]{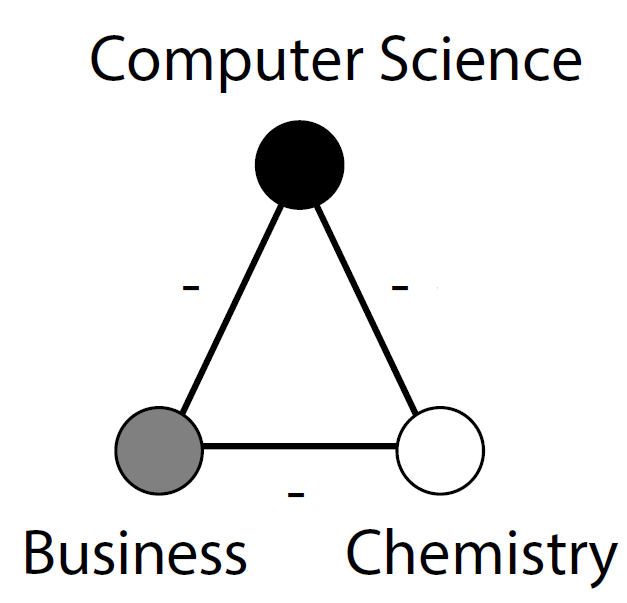}
\label{fig:theseauthors-neg}
}
\end{tabular}
\end{center}
\caption{\label{fig:theseauthors} Triangles associated with the country of origin, gender, and educational background.} 
\end{figure}
\end{example}

A long line of relevant literature has established the use of triangles to model social structures~\citep{cartwright56structural,easley10networks,heider58psychology,morrissette67no}. In our own setting, the triangle represents the fundamental building block of our faultline measure, as any structure that includes more members (e.g. a rectangle) can be trivially modeled via (or broken down to) triangles.
The figure reveals the existence of 3 possible types of triangles among the members of the team, according to the signs on their edges:
$(+,+,+)$, $(-,-,-)$, and $(-,-,+)$. By definition, $(+,+,-)$ triangles cannot exist as they would imply that 
2 individuals have the same value as the third one but not the same as each other.  
We observe that faultlines can only appear in the presence of $(+,-,-)$ triangles that consist of one positive
and two negative edges, such as the one for the \emph{country of origin} feature shown in Fig~\ref{fig:theseauthors}(a). 
Given that faultlines can only emerge in the presence of $(+,-,-)$ triangles, we refer to these as \emph{Conflict Triangles}. 

A conflict triangle captures the intuition that two people from the same country are more likely to interact with each other than to the third person, thus enabling the creation of a potential faultline.
On the other hand, A faultline could never occur for the \emph{gender} feature (Fig~\ref{fig:theseauthors}(b)), as all three authors have the same value (\emph{Male}). Similarly, since all three authors have a
different value for the \emph{undergraduate major} feature (Fig~\ref{fig:theseauthors}(c)), there is no faultline potential. This is consistent with faultline theory, which states that 
faultlines cannot emerge in the presence of perfect homogeneity or perfect diversity~\citep{lau1998demographic,gratton07bridging}. 

\begin{figure}[htb]
\begin{center}
\begin{tabular}{c}
\subfloat[High faultline potential]{
\includegraphics[scale=.15]{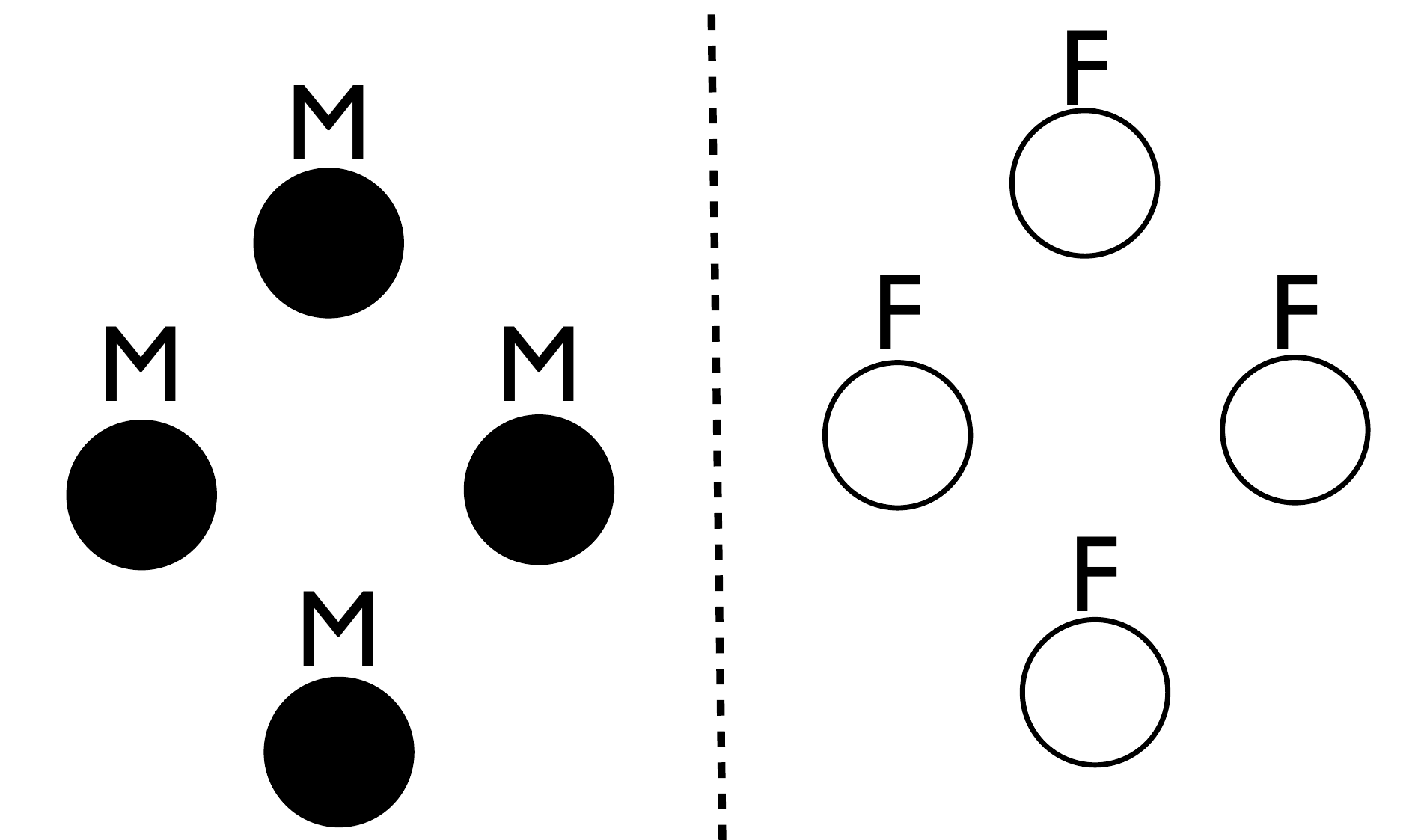}
\label{fig:groups1}
} \hspace{5ex}
\subfloat[No faultline potential]{
\includegraphics[scale=.15]{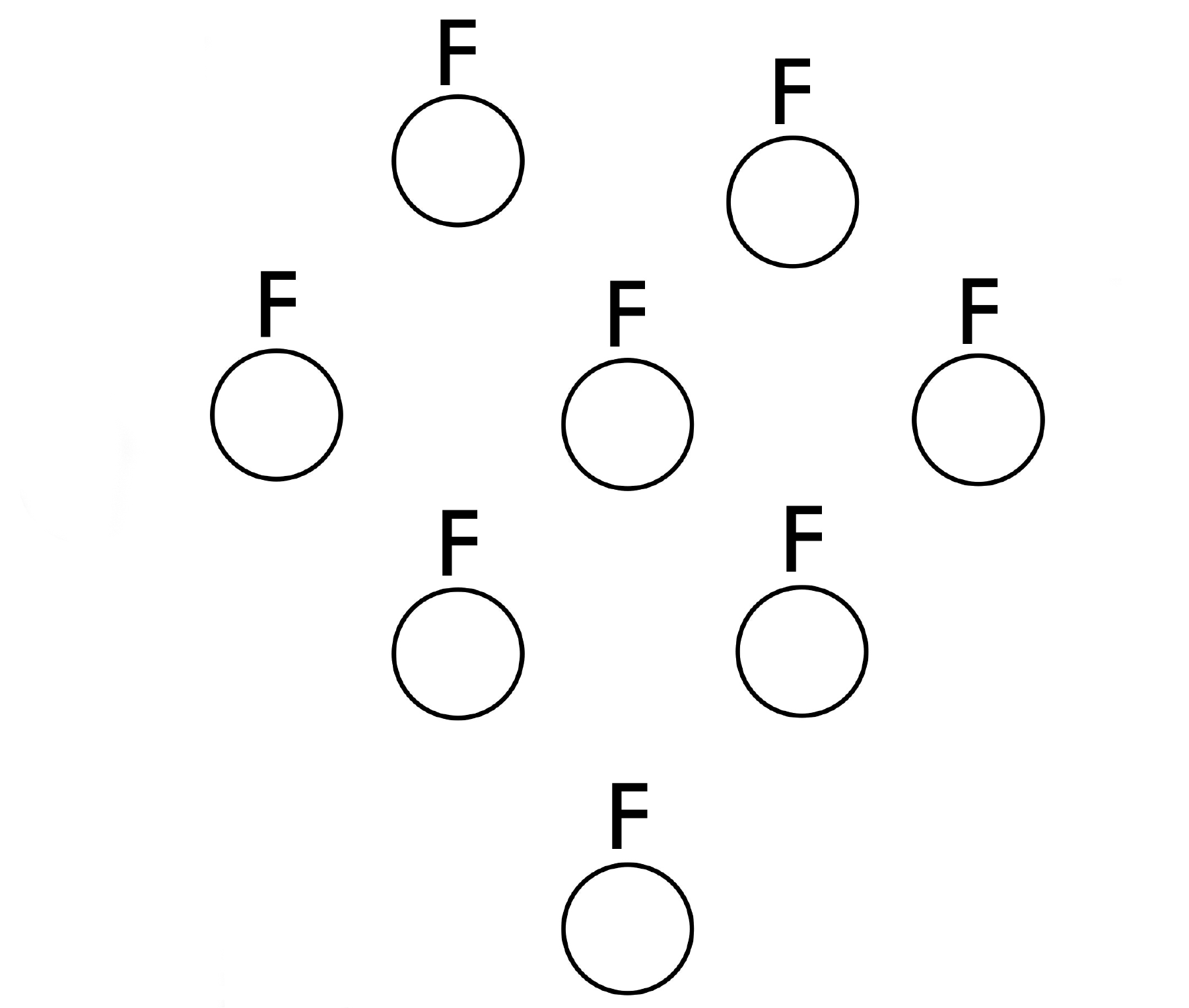}
\label{fig:groups2}
}\hspace{5ex}
\subfloat[No faultline potential]{
\includegraphics[scale=.15]{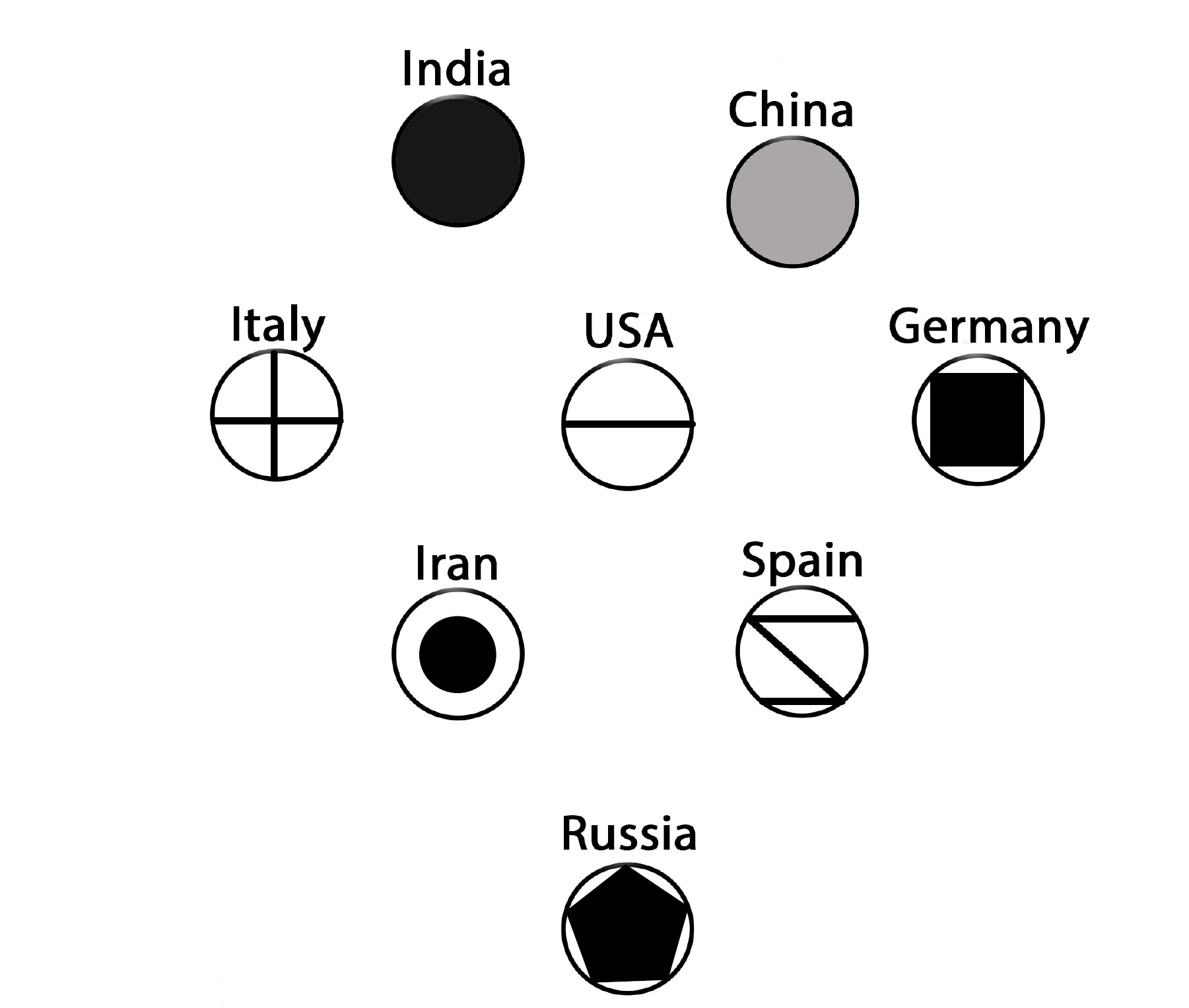}
\label{fig:groups3}
}
\end{tabular}
\end{center}
\caption{\label{fig:groups} Examples of teams with high and low faultline potential.}
\end{figure}

The ability of triadic relationships to capture the perfect homogeneity/diversity principles that are mandated by faultline theory maintaints its usefulenes in a team-formation setting. Consider the example in Figure~\ref{fig:groups1}. The team in the figure represents the worst-case scenario in terms of faultline potential for the gender feature: a 50-50 split between two large homogeneous groups of males (M) and females (F). Figure~\ref{fig:groups2} shows an example of a team with no faultline potential for gender, as it consists exclusively of female members. Even though increased homogeneity is indeed one of the ways to reduce faultline potential, it is wrong to equate diversity with the emergence of faultlines. We demonstrate this in Figure~\ref{fig:groups3}. All the members of the teams in this figure have different
values with respect to the feature \emph{country of origin}. We observe that, as in cases of perfect homogeneity, faultlines cannot exist in the presence of perfect diversity.  This observation reveals that the task of measuring a team's faultlines goes beyond simply measuring its diversity with respect to different features. Similarly, a team formation algorithm has to carefully balance the two states of homogeneity and diversity within a team in order to achieve a low potential for faultlines.

\subsection{Feature Alignment:}
The next essential step toward the design of a triangle-based faultline measure is the consideration of \emph{the alignment} of conflict triangles across multiple features~\citep{meyer2013team}.
Consider three individuals $(i, j, k)$ defined within a space of features $\mathcal{F}_{\mathcal{T}}$. Given a feature $f\in\mathcal{F}_{\mathcal{T}}$, 
let $\tau=<(i,j),k>$ be a conflict triangle such that $\worker_i(f)=\worker_j(f)$ \emph{and} $\worker_i(f),\worker_j(f)\neq \worker_k(f)$.
Let $iscon(\tau,f)$ be a function that returns 1 if $\tau$ is a conflict triangle for $f$ and 0 otherwise.

If the same conflict triangle emerges for a second feature $f'$, we say that $\tau$ is \emph{aligned} across the two features $f$ and $f'$ (i.e. $iscon(\tau,f)=iscon(\tau,f')=1$). 
Let $p(\tau,T)$ return the percentage of all available features of team $T$ for which $\tau$ is aligned (i.e. for which $\tau$ appears as a conflict triangle). Formally:
$$
p(\tau,T)=\frac{|\{f\in\mathcal{F}_{\mathcal{T}}:\ iscon(\tau,f)=1\}|}{|\mathcal{F}_{\mathcal{T}}|}
$$

We say that a triangle $\tau$ from team $T$ is \emph{fully aligned} if it is aligned across all team features (i.e. $p(\tau,T)=1$).
Then, we define the faultline potential of a given team $T$ as follows:

\begin{equation}
\label{eq:act}
CT(T)=\sum_{\tau\in\mathcal{D}_T}p(\tau,T) 
\end{equation}

where $\mathcal{D}_T$ is the set of all distinct conflict triangles $<(i,j),k>$ that appear across any of the features in $T$. Our measure has a probabilistic interpretation, as it encodes the expected number of successes (conflict triangles) that we would get after $|\mathcal{D}_T|$ Bernoulli trials, where each trial
corresponds to a different $\tau\in\mathcal{D}_T$ and has a success probability equal to $p(\tau,T)$. The trial for conflict triangle $\tau$ involves sampling )(uniformly at random) a feature $f$ from $\mathcal{F}_{\mathcal{T}}$ and is successful if $\tau$ is a conflict triangle for $f$. Hence, a perfectly aligned triangle would succeed for any sampled feature and would increment the team's score by 1. Similarly, the trial for a triangle $\tau$ that is aligned over half of the team's features would have a $50\%$ of success and would increment the team's score by $0.5$.

The penalty that Eq.~(\ref{eq:act}) assigns to each conflict triangle in the team is directly proportional to the triangle's alignment across the team's features.

Under this definition, the minimum faultline potential is assigned to perfectly homogeneous or perfectly diverse teams, as they both
include zero conflict triangles. On the other hand, in accordance with faultline theory~\citep{lau1998demographic}, the maximum faultline potential is assigned to teams that can be split into two perfectly homogeneous subgroups of equal size.

\spara{Learning the appropriate penalization scheme from real data:}
The definition given in Equation~\ref{eq:act} intuitively applies, for each conflict triangle, a penalty that is directly proportional to the triangle's alignment across the team's features. We thus expect it to be a reasonable modeling choice for many domains. However, in practice, this penalization scheme may not be appropriate for
a specific domain or application. Therefore, we extend our framework via by describing a methodology that allows practitioners to \emph{learn} the appropriate penalization on function for their domain, based on information from existing teams in the same domain. We present the details of our technique for learning the penalization parameters in Section~\ref{sec:alignments}.

\subsection{Efficiently computing a team's faultline potential}
\label{sec:eff}
The computation requires us to count the total number of conflict triangles
across all features. Thus, for $T\subseteq W$, $\calF(T)$ can be computed in
polynomial time.  For this, one has to consider all triangles appearing in
the feature graphs and count how many of those are conflict triangles.
The running time of the naive computation is $O(m |T|^3)$ where $|T|$ is
the size of the team and $m$ is the number of
features. Next, we present a method for significantly speeding up this computation.

Given a set of workers $T\subseteq W$, 
 and a feature $f$ that takes values $v_1,\ldots,v_L$, we summarize the values of
$f$ observed among the workers in $T$ via the \emph{aggregate feature vector}
$\aggf(T,f)$ such that $\aggf(T,f)[v_j]$ gives the number of workers in $T$
that have a value equal to $v_j$. 
We observe that these aggregate vectors can be
computed in $O(m|T|)$ time by simply counting all feature values of all workers.
Once the aggregate feature values have been computed, the faultline potential for each
feature $f$ that takes values $v_1,\ldots,v_L$ can be written as follows:

\begin{equation}\label{eq:fast_faultline}
\calF(T,f) = \sum_{j=1}^N {\binom{\aggf(T,f)[v_j] }{2}} \left(|T| - \aggf(T,f)[v_j]\right)
\end{equation}

We observe that, for any feature $f$ with $L$ different possible values, the faultline potential with 
respect to $f$ can be computed in $O(L)$ time using the above equation. Thus, the overall
faultline potential $\calF(T)$ can be computed in $O(mL)$. Given that both the number of
features $m$ and the number of possible values for each feature $L$ are usually small 
constants, this computational cost is negligible compared to the time required to
create the aggregate feature values. The use of the aggregate feature vectors also allows us to update
the score in constant time, as required by the second efficiency principle of faultline-aware team-formation. Specifically,
if an individual $i$ joins or leaves the team, we only need to update (in $O(m)$) the number of conflict triangles that are due to the aggregate counts that change due to the addition or removal of $i$.

\section{The {\partition} Problem}\label{sec:problem}
In this section, we formally define the {\partition} problem, i.e., the problem of
partitioning a set of workers $W$ into
$\ell$ teams of equal size such that the total faultline potential score across
teams is minimized. We show that this problem is not only NP-hard to solve,
but also NP-hard to approximate within any bounded approximation factor, unless $\text{P}=\text{NP}$.
Then, in Section~\ref{sec:alg}, we present an efficient heuristic algorithm for its solution.

First, we extend the notion of faultline potential to a collection of teams.
For any partitioning $\bT = \{T_1, T_2, \cdots, T_\ell\}$ of workers into
$\ell$ teams, we use $\calF(\bT)$ to denote the total faultline potential of
all teams in $\bT$. Formally: 
\begin{equation}
\label{eq:part}
\calF(\bT) = \sum_{i=1}^\ell \calF(T_i).
\end{equation}

We can thus define the {\partition} problem as follows:

\begin{problem}[{\partition}]
Given a pool of workers $\workers$ (with $|\workers| = \ell \times k$), 
find a partitioning $\bT = \{T_1, T_2, \cdots, T_\ell\}$ of the workers $W$
into $\ell$ teams of size $k$ such that $\calF(\bT)$ is minimized.
\end{problem}

Next, we proceed to analyze the hardness of the {\partition} problem. Our results apply for the more general problem of partitioning
a population into teams with specific but possibly different sizes.

\begin{theorem}
\label{theo:partition}
The {\partition} problem is NP-hard to solve.
\end{theorem}

Theorem~\ref{theo:partition} implies that the {\partition} problem cannot be optimally solved in polynomial time unless $NP=P$. Next, we provide a formal proof of this theorem.

\begin{proof}
We present a polynomial-time reduction from the NP-Complete {\sc k-Clique\ Partitioning} problem
to our {\partition} problem~\citep{gary1979computers,rosgen2007complexity}. The {\sc k-Clique\ Partitioning}
is a decision problem which asks the following question: Given a graph
$H=(V,X)$, is it possible to partition the nodes of the graph into
disjoint cliques of size $k$?

Given a graph $H=(V,X)$ (with $V$ nodes and $X$ edges), we first create the
complement of $H$ denoted by $H'=(V, X')$. Clearly, any clique of size $k$
in the original graph $H$ corresponds to a set of $k$ nodes with
no edges among them in $H'$.

For our reduction, every node $i\in V$ will correspond to a worker for our
problem. Also, we will interpret each edge in $H'$ as an agreement (``+'') and each
missing edge as a disagreement (``-'').
Then, for every edge $(i,i')$ in $H'$, we create a feature $f_{(i,i')}$ and then
construct the corresponding feature graph $G_{f_{(i,i')}}$ that contains one
positive edge connecting nodes $i$ and $i'$, while all other edges, connecting all
pairs of nodes, are negative. Fig~\ref{fig:reduction} shows how an example
graph $H$ with three edges is transformed into three feature graphs.

\begin{figure}[ht]
\begin{center}
\includegraphics[scale=.2]{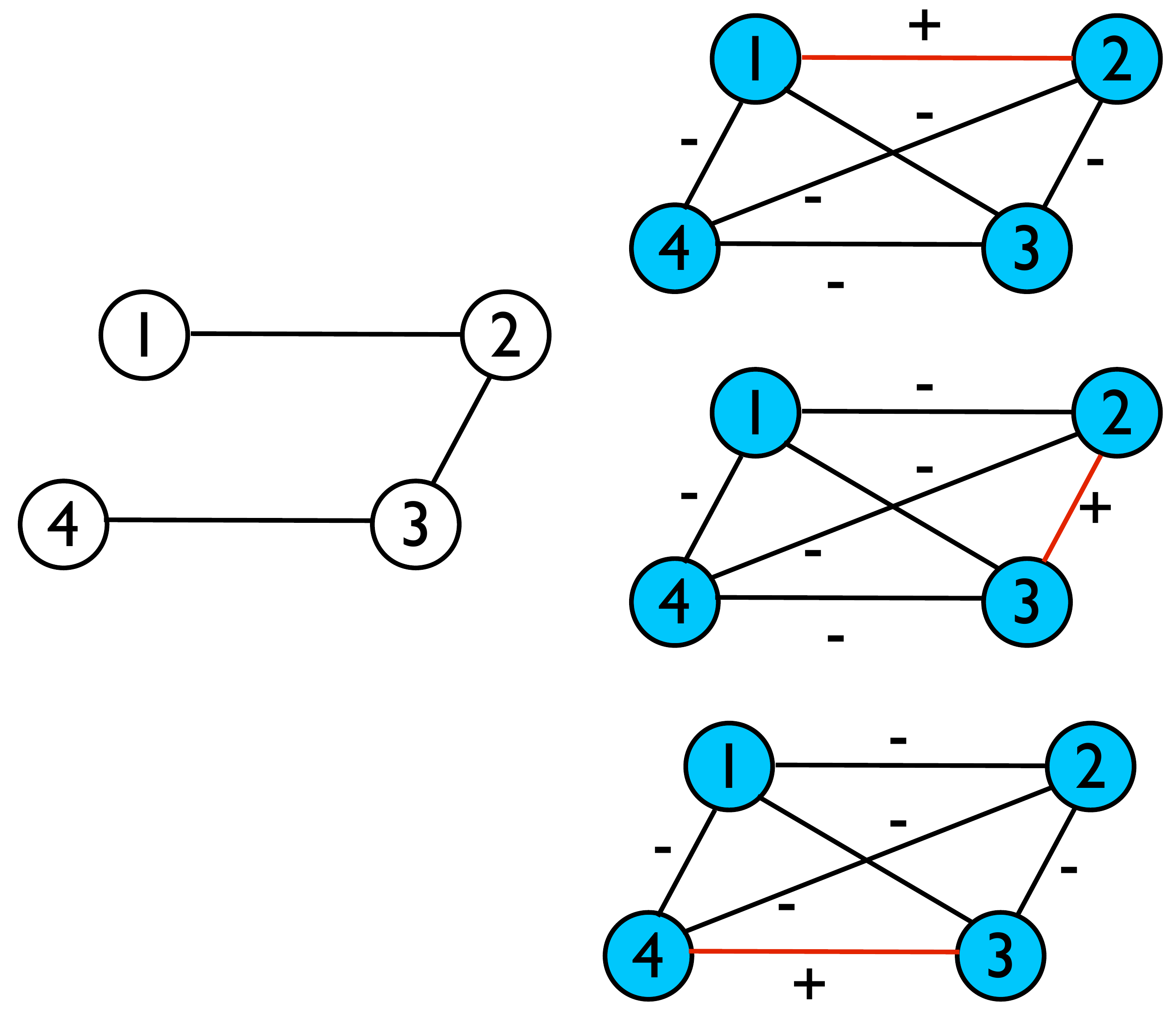}
\end{center}
\caption{\label{fig:reduction}Graph $\widehat{H}$ (in gray) and its feature graphs for the 
corresponding {\partition} problem}
\end{figure}

Now consider the optimal solution to this instance of the {\partition}
problem. Since the size of each team is fixed ($k$), it is easy to see
that each edge of $H'$ that falls within one team creates $(k - 2)$ conflict
triangles. This implies that the optimal solution is the one that minimizes
the total number of edges that fall within the partitions. More specifically,
the optimal solution has a faultline potential equal to zero if and only
if there exists a partitioning of the nodes in $H'$ with no edge inside 
the partitions which further corresponds to a partitioning of the nodes
in $H$ into cliques.
\end{proof}

\begin{corollary}
\label{theo:partitionapprox}
The {\partition} problem is NP-hard to approximate within any factor.
\end{corollary}

\begin{proof}
We will prove the hardness of approximation of {\partition} by contradiction.
Assume that there exists an $\alpha$-approximation algorithm for the
{\partition} problem. Then if $\bT^{\ast} = \{T_1^{\ast}, T_2^{\ast}, \cdots,
T_\ell^{\ast}\}$ is the partitioning with lowest faultline potential and
$\bT^{\calA} = \{T_1^{\calA}, T_2^{\calA}, \cdots, T_\ell^{\calA}\}$ is
the solution output by this approximation algorithm, it will hold that
$\calF(\bT^{\calA})\leq \alpha \calF(\bT^{\ast})$.
If such an approximation algorithm exists, then this algorithm
can be used to decide the instances of the
{\sc k-Clique\ Partitioning} problem, for which the optimal solution
has a faultline potential equal to 0. However, this contradicts
the proof of Theorem~\ref{theo:partition}, which indicates that these
problems are also NP-hard. Thus, such an approximation algorithm does not exist.
\end{proof}

\subsection{The {\matching} algorithm}
\label{sec:alg}
In this section, we present an algorithm for the {\partition} problem. 
We refer to the algorithm as {\matching} and provide the pseudocode 
in Algorithm~\ref{algo:iterative}. The Python implementation of the algorithm is available online~\footnote{\url{https://github.com/sanazb/Faultline}}.

The algorithm starts
with a random partitioning of the input population into $\ell$ equal-size groups and then reassigns individuals to teams in an iterative fashion until the faultline potential
of the obtained partitions does not improve across iterations.

\begin{algorithm}
\begin{algorithmic}[1]
\Statex {\bf Input:} Set of workers $W$ with $m$ features and the number of
desired partitions $\ell$.
\Statex {\bf Output:}  Partitioning $\bT = \{T_1,T_2,\ldots, T_\ell\}$ 
\State Randomly partition $W$ into $\bT = \{T_1,\ldots,T_\ell\}$
\While { $\calF(\bT)$\text{ has not converged} }
\State $\bc = $  {\assigncosts} $(W,\bT)$
\State $\bT$ = {\reassignteams}$\left(\bT,\bc\right)$
\EndWhile
\State return $\bT$ 
\end{algorithmic}
\caption{\label{algo:iterative}The {\matching} algorithm.}
\end{algorithm}

In each iteration, the algorithm starts with a partitioning of the set $W$ into $\ell$
groups and forms a new assignment with (ideally) a lower faultline potential score.
This is done by executing two functions: {\assigncosts} and {\reassignteams}. The
{\assigncosts} function returns a cost associated with the assignment
of every individual to every team; i.e., $\bc(i,T_j)$ is the cost of assigning individual
$i$ into team $T_j$. These costs are used by {\reassignteams} to produce
a new assignment of individuals to teams -- always guaranteeing that the teams are of equal
size. Next, we describe the details of these the two main routines of {\matching}.

\spara{The {\assigncosts} routine:}
This routine, assigns to every worker $i$ and team $T_j$
cost $\bc(i,T_j)$, which is the cost of assigning worker $i$ to team $T_j$.
In order to compute these costs,
{\assigncosts} considers the
current teams in $\bT$ as a baseline to evaluate if the assignment of worker
$i$ to team $T_j$ can lead to fewer 
conflict triangles. Thus, an intuitive definition 
of cost is the number of conflict triangles 
that $i$ incurs when he joins $T_j$. This is equal to 
$\calF(T_j \cup \{i\})$ if $i \not \in T_j$ and  
$\calF(T_j)$  if $i \in T_j$.  

We observe that,
if worker $i$ already belongs to team $T_j$, the reassignment is not going to change the
size of the resulting team. However, if 
$i \not \in T_j$  the assigning $i$ to $T_j$ creates a team of size 
$(k + 1)$.
This is problematic, since the number of conflict triangles in teams of size $k$ is not comparable to that in teams of size $(k+1)$. 
This can be resolved by introducing a normalization factor which measures the 
\emph{maximum}
possible number of conflict triangles in a team of a fixed size. 
Formally, for a team of size $k$, we use $\Delta_k$
to denote the maximum possible number of conflict triangles that can emerge 
in the team across all features.
Now, we compute the cost function as follows:
\begin{equation}\label{eq:update2}
\bc(i, T_j) =
\begin{cases} 
\calF(T_j \cup \{i\}) / \Delta_k & \text{if\ } i \not \in T_j \\ 
\calF(T_j) / \Delta_{k+1} & \text{if\ } i \in T_j  
\end{cases}
\end{equation}

\emph{Running time:}
Note that computing all three cost functions
can be done in $O(m|W|)$ using the aggregate feature vectors as discussed
in Section~\ref{sec:eff}.

\spara{The {\reassignteams} routine:} {\reassignteams} takes as input a current
a cost of assigning each one of the $n$ individuals 
into each one of the $\ell$ teams and outputs a new partition of the 
individuals into $\ell$ equal-size groups.
The algorithm, views this partitioning problem as a minimum weight $b$-matching problem~\citep{burkard12assignment}
in a bipartite graph, where the nodes on the one side correspond to 
$n$ individuals and the nodes
on the other side correspond to $\ell$ teams. In this graph, 
there is an edge between every 
individual $i$ and team
$T_j$. The weight/cost of this edge is, for example, 
$\bc(i,T_j)$ -- computed as described above.
Finding a good partition
then translates into  picking a subset of the edges of the bipartite graph,
such that the selected edges have a
minimum weight sum,  every individual in the subgraph defined by the
selected edges has degree $1$, and each team has degree $k$. This would mean that every
worker is assigned to exactly one cluster and every cluster has exactly $k$ members.
 This is a classical $b$-matching problem that can be solved
in polynomial time using the Hungarian algorithm~\citep{burkard12assignment,kuhn55hungarian}.

\spara{Variable-size partitioning:} It is important to point out that our algorithm can be easily modified to partition a population into teams of fixed but possibly different sizes. The  {\reassignteams} routine in our algorithm computes a new assignment of individuals to teams by solving a minimum weight b-matching problem in a bipartite graph where nodes on the right represent individuals and nodes on the left represent the available spots/positions in each team. This setup gives us the flexibility to choose the number of available spots in each team. In fact, this is how the algorithm enforces equal-size teams in our current implementation.

\spara{Computational speedups:} Computing the new partition using the Hungarian algorithm,
requires $O(n^3)$ time. This is a computationally expensive operation, especially since this step needs to be completed in each iteration of 
{\matching}.
In order to avoid this computational cost, we solve the bipartite $b$-matching problem
approximately using
a greedy heuristic that works as follows:
in 
each iteration the edge $(i,T_j)$ with the lowest cost $\bc(i, T_j)$
is selected, and worker $i$ is assigned to the $j$-th team $T_j$;
this assignment only takes place if: $1)$ worker
$i$ is not assigned to any team in an earlier iteration, and $2)$ the $j$-th
team has less than $k$ workers so far (i.e., if it has not reached the desired
team size). This is repeated until all the workers are assigned to a team.

To find the minimum cost edge in each iteration we need to
sort all edges with respect to their costs and then traverse them 
in this order. Since there are $O(n\ell)$ edges, the running time of this greedy alternative
is  $O(n \ell \log (n \ell))$ per iteration. 

\section{Experiments}\label{sec:experiments}
In this section, we describe the experiments that we performed to evaluate our methodology. 

\subsection{Datasets}
\label{sec:datasets}
\bpara{\underline{\adult}:} The \adult\ dataset is a census dataset from UCI's machine learning 
repository. It contains information on $32,561$ individuals;
the features in the data are \emph{age, work class,
education, marital status, occupation, relationship, race, sex, capital-gain, capital-loss, hours-per-week,}
and \emph{native country}  \footnote{\url{https://archive.ics.uci.edu/ml/datasets/Adult}}. 
We convert non-categorical features to categorical features as follows:
for \emph{age} and \emph{hours-per-week}
we bin their values into buckets of size $10$. 
Also, we convert both \emph{capital-gain}
and \emph{capital-loss} into binary features depending whether 
their value is equal to zero or not.

\bpara{\underline{\census}:} The \census\ dataset is extracted from the US government's "Current
Population Survey'' \footnote{\url{http://thedataweb.rm.census.gov/ftp/cps_ftp.html}}. We focused
on the most recent collected data from the year $2014$. Our dataset contains 
census information on
$200,469$ individuals. The dataset includes the following features: \emph{marital status, gender, education,
race, country, citizen,} and \emph{army}.

\bpara{\underline{\dblp}:} The \dblp\ dataset 
is created by using the latest snapshot of the 
DBLP website and filtering only authors that published papers on tier-1 and tier-2
computer science (NLP, IR, DM, DB, AI, Theory, Networks) 
conferences and journals \footnote{\url{http://webdocs.cs.ualberta.ca/~zaiane/htmldocs/ConfRanking.html}}.
Although the only known attribute in the raw dataset is the 
\emph{country of origin},
we extracted the following features for each of the $57,972$ authors, based on their publications: \emph{number of years active}, 
\emph{primary area of focus} (based on number of publications),\emph{average
number of publications in ten years}, and \emph{total number of
publications}. We also computed a \emph{quality} feature for each 
author, by giving her 2 points for each paper published in a top-tier conference and 1 point for all other papers. 
We bin both the \emph{total number of publications} and the \emph{average number of publications} into buckets of size $10$, and bin the \emph{quality} score into buckets of size $5$.

\bpara{\underline{\bia}:}
This dataset is collected from entry surveys taken by all students who take the Analytics course offered by one of the authors of this paper.
The data was collected during 6 different semesters and includes data from 502 graduate students. It consists of 85 teams, with an average of 5.9 students per team. For each student,
the dataset includes the major of the degree they were pursuing at the time of the data collection, the major of their bachelor's degree, gender, country,  and a self-assessment of her level with respect to machine learning, analytics, programming, and experience with team projects. The assessments
are given on a scale from 0 (no experience) to 3 (very experienced). 
For each team, we also have its performance (on a scale of 0 to 100) on a collaborative, semester-long project that accounts for $70\%$ of the entire grade, as well as the average satisfaction level (on a scale of 0 to 7) of the team's members with the way the team operated. For each team we computed tension (bad triangles) for each team across all features.

\mpara{\underline{\synthetic}:}
In order to control the number of conflict triangles in our data, we have developed a method
to create synthetic datasets given a target percentage of conflict triangles.
First, we assume that our pool of workers $W$ is going to consist of a 
single  feature which
can only take $3$ different values 
$X$, $Y$, and $Z$. Let's define $x$, $y$, $z$ to be the 
number of data points with these values respectively. 
Now, it is clear that $\calN(W) = x \times y 
\times z$. On the other hand, given that total number of workers is $n$ we have $x + y + z = 
n$. Note that if the value of $x$ is given, we can use these equations to compute the value of 
$y$ and $z$ as well. To create our datasets, we try different values of $x$ and then 
we solve for variables $y$ and $z$. Then, we 
randomly partition workers into three groups of size $x$, $y$, and $z$ and assign
the value $X$, $Y$, and $Z$ to them respectively.

\mpara{\underline{\synthetictwo}:} In order to compare different faultline measures --{\asw}, Subgroup Strength ({\sss}),
and our {\ct} measure-- we generate a dataset as follows. We consider three features: \emph{Race (Asian, White, Black, Native American)},
\emph{Country (USA, China, England, France)}, and \emph{Education (High-school, Undergraduate, Graduate)}. Then, given a team size
$TS$ and a number of subgroups $SN$, we generate $100$ teams that include $TS$ individuals divided into $SN$ completely homogeneous
subgroups.  Within each subgroup, all individuals have the same value for each feature $F$. This value $V$ is selected with 
a probability that is inversely proportional to the number of subgroups in the team that has already been assigned $V$ for this feature. This process allows us to create perfectly homogenous groups that are highly dissimilar from each other. We repeat the process
for $TS\in\{4, 8, 16, 32, 64\}$. Given a value for $TS$, we start with $SN=1$ (a perfectly homogeneous team) and double the value until $SN=TS$ (one individual per subgroup). For instance, for $TS=16$, we consider $SN\in\{1,2,4,8,16\}$. This process generates a
total of $3100$ teams. Controlling the number of perfectly homogenous subgroups allows us to control diversity and simulate multiple scenarios of conflict between different types of subgroups within the team.

\begin{table*}
  \centering
  \topcaption{Statistics for the real datasets.}
  \begin{tabular}{@{}lccccc@{}}
  \toprule  
  \textbf{Dataset} & \textbf{Size} & \textbf{Features} & \% of conflict triangles\\
  \midrule
  \midrule
  \dblp & 57,972 & 6  & 35\% \\
  \adult & 32,561 & 12 & 41\% \\
  \census & 200,469 & 7 &  44\% \\
  \dblpAug & 155 & 9 & 47\% \\
  \bia & 502 & 8 & 62\% \\
  \synthetic & 400 & 8 & 8\%\\
  \bottomrule
  \end{tabular}
  \label{table:freq}
\end{table*}

\mpara{Discussion:} 
Table~\ref{table:freq} shows some basic statistics for our datasets. As mentioned earlier,
the {\synthetic} dataset allows us to tune the percentage of different types of triangles.
The synthetic instance reported in Table~\ref{table:freq} corresponds to a dataset of size
$400$ with $8$ features where we set the percentage of negative and positive triangles
to $8\%$ and $25\%$ respectively.
Fig~\ref{fig:corr} illustrates the \emph{Cramer's V} values for all 
pairs of features in all datasets. Cramer's V value
is a standard measure the correlation between
two categorical variables~\citep{cramer2016mathematical}. It has a value of $1$ when two variables are perfectly correlated
and $0$ if there is absolutely no correlation.
The figure illustrates that {\adult} and {\census} are similar in terms of feature correlation. Specifically,
we observe a small correlation for the majority of the features
and only a couple of them with high correlations. On the other hand, {\dblp} exhibits significantly
higher correlation patterns. 

\begin{figure*}
\centering
\subfloat[{\dblp}]{
\includegraphics[scale=.35]{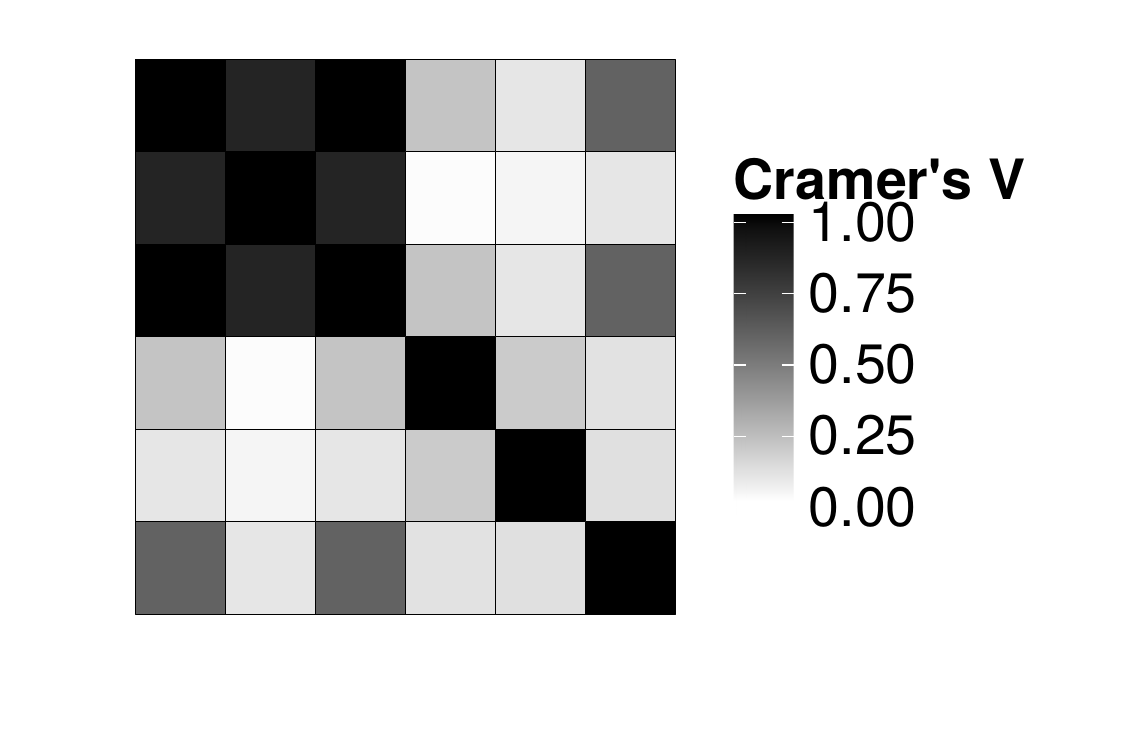}
\label{fig:dblp_corr}
}
\subfloat[{\adult}]{
\includegraphics[scale=.35]{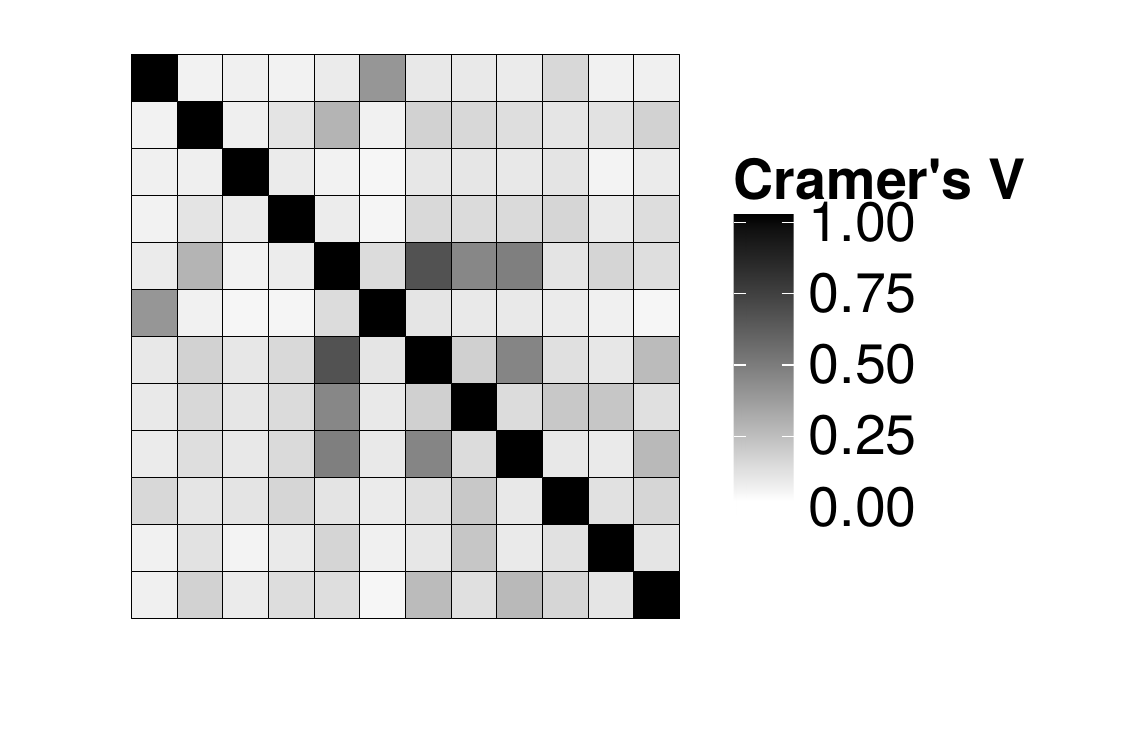}
\label{fig:adult_core}
}
\subfloat[{\census}]{
\includegraphics[scale=.35]{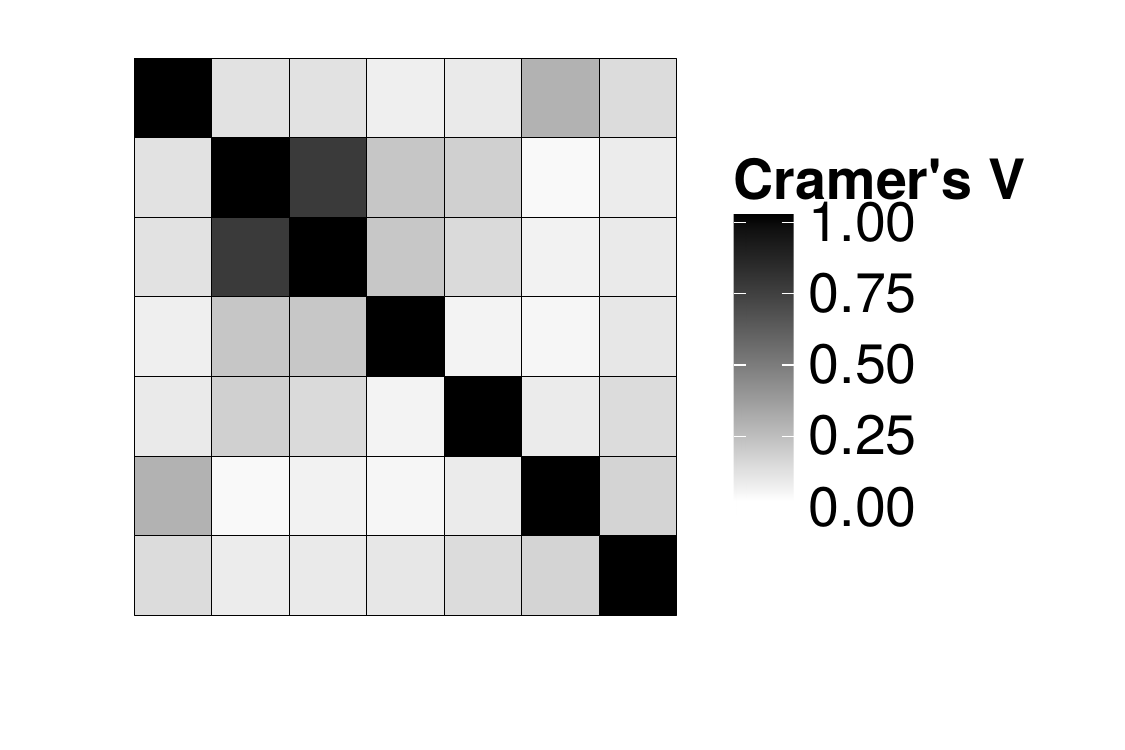}
\label{fig:census_core}
}  \\
\subfloat[{\synthetic}]{
\includegraphics[scale=.35]{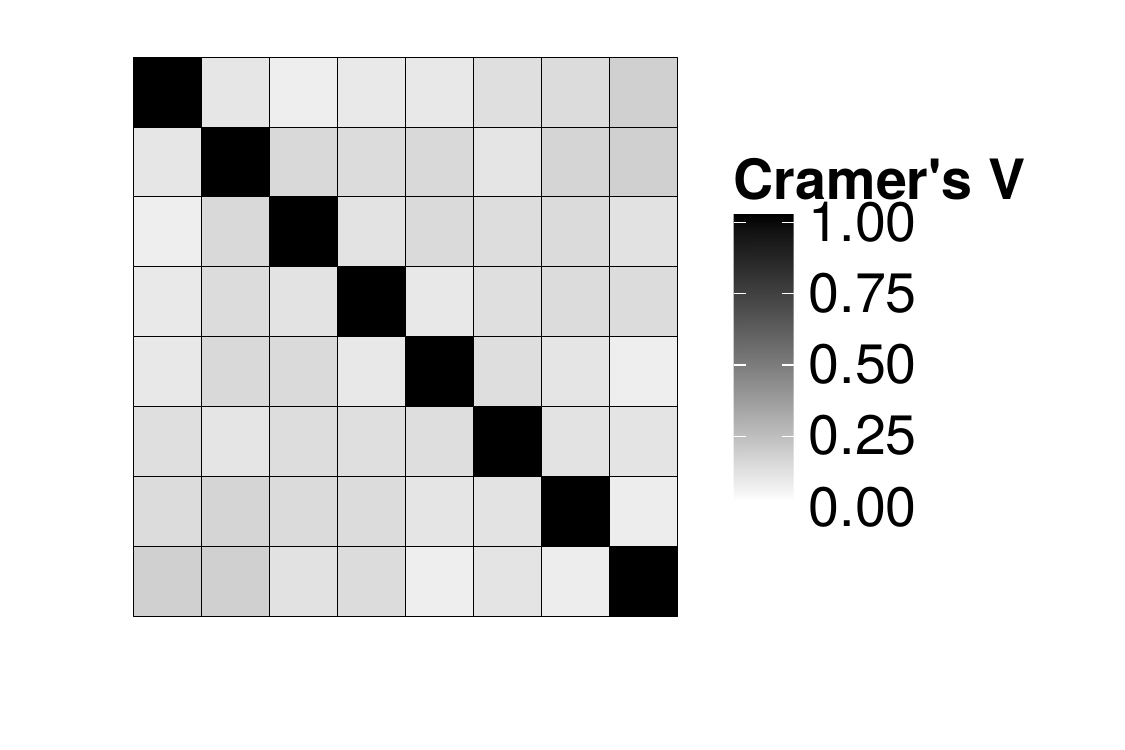}
\label{fig:synthetic}
} 
\subfloat[{\bia}]{
\includegraphics[scale=.35]{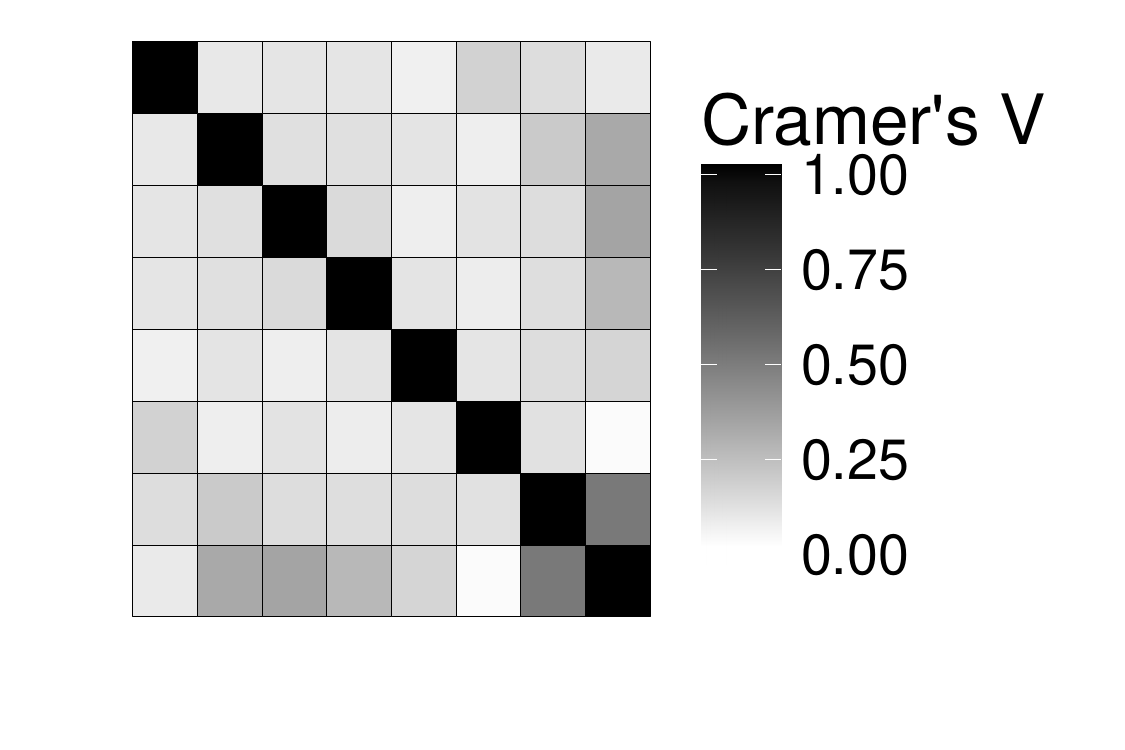}
\label{fig:bia_core}
} 
\subfloat[{\dblpAug}]{
\includegraphics[scale=.35]{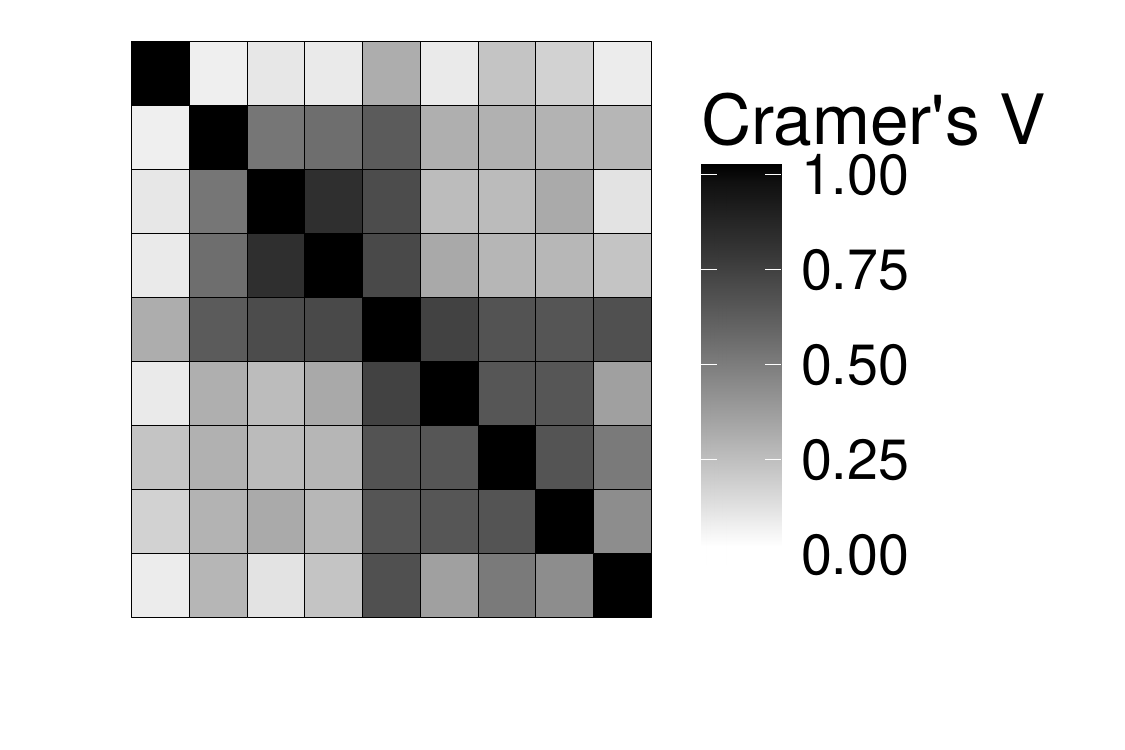}
\label{fig:users_core}
} 
\caption{\label{fig:corr} The Cramer's V values for all pairs of features for all datasets.}
\end{figure*}

\subsection{Evaluation on the {\partition} problem}
In this section, we evaluate the performance of our algorithms for the\\ {\partition} problem.

\mpara{\underline{Baselines:}} We compare our {\matching} algorithm with two baselines: {\greedy} and {\dya}.
The {\greedy} algorithm takes an iterative approach that creates a single team in each iteration and thus it requires
$\ell$ iterations to create all $\ell$ teams. Each team is constructed as follows.
First, the algorithm selects two random workers. It then continues by greedily adding the worker that minimizes
the faultline score of the team. Once the size of team reaches $k$, the algorithm removes the selected members from the pool of experts
and moves on to build the next team. Finally, {\dya} is a clustering algorithm that tries to create
equal-size partitions such that the number of positive (negative) edges within the teams is maximized (minimized)\texttt{~\citep{malinen14balanced}}.  

\mpara{Evaluation metric:} 
for every algorithm, we measure its performance via the faultline potential of the set of teams that it creates, as per Equation~\ref{eq:part}. Because some of our comparisons require plotting results obtained from datasets of different sizes in the same
figure, we apply the following dataset-specific normalization. For a dataset of size
$n$, we divide the faultline potential
of a partitioning obtained for this dataset with the 
the total number of triangles that can be encountered in datasets of this size, 
i.e., ${\binom{n}{3}}$. Thus, the $y$-axis
of all our plots is in $[0,1]$.

\subsubsection{Varying the population size} For each dataset, we randomly select, with replacement,
$100$ sets of $n$ individuals, for $n\in\{100,200,400,800,1600\}$. We then use the
algorithms to partition each set into teams of size $5$. For each algorithm,
we report the average faultline potential achieved over all sets
for every value of $n$, along with the corresponding $90\%$
confidence intervals. The results for all three datasets are shown in
Figure~\ref{fig:population_size}. We also report the computational time (in seconds) of each algorithm for each value of $n$ in 
Figure~\ref{fig:runtimes}.

\begin{figure*}
\centering
\includegraphics[scale=.4]{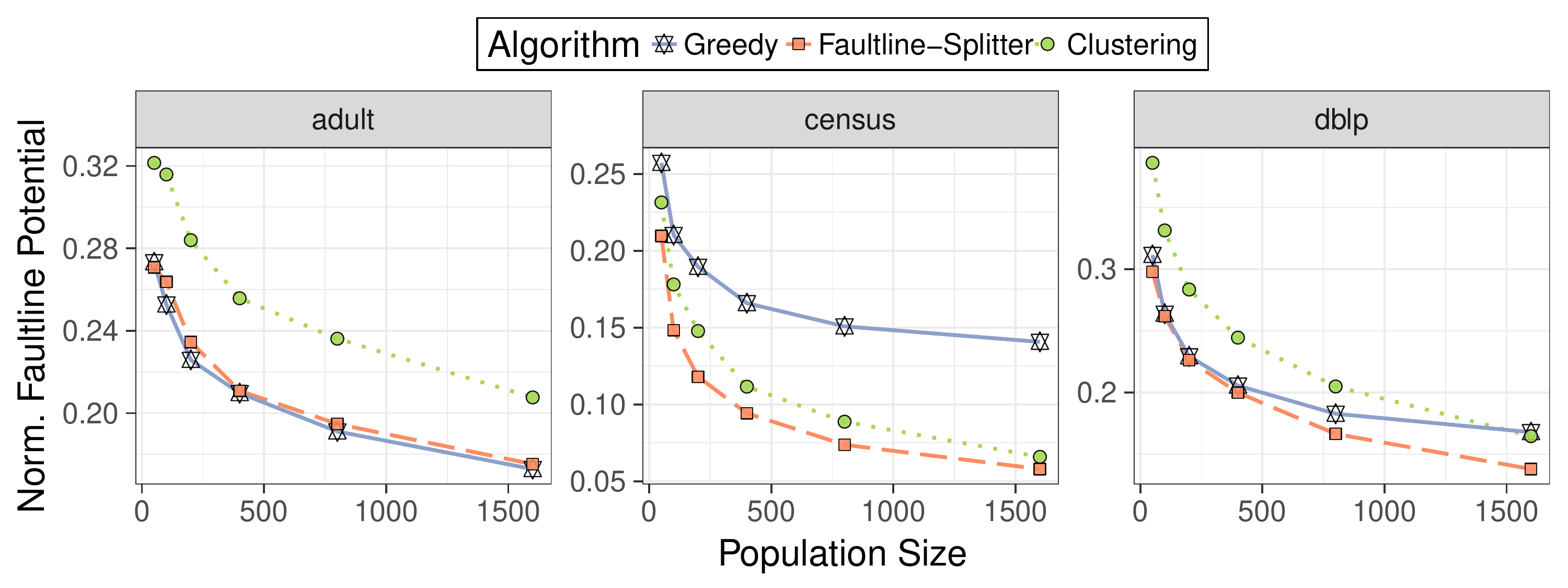}
\caption{\label{fig:population_size} Faultline results for different population sizes (parameter $n$)}
\end{figure*}

The first observation is that all the algorithms perform better as the size of the population
increases, with the achieved normalized faultline potential values ultimately converging to a low value around $0.1$, for all datasets. An examination of the data reveals that we can confidently attribute this trend to the fact that increasing the size of the population leads to the introduction of identical or highly similar individuals (i.e. in terms of their feature values). This makes
it easier to form low-faultline teams. This is not a surprising finding in real datasets, which tend to include large clusters of similar points, rather than points that are uniformly distributed within the multi-dimensional space defined by their features.

We observe that The {\matching} algorithm consistently achieves the best results across datasets, while the
the {\greedy} algorithm outperforms {\dya} in two of the three datasets {\dblp} and {\adult}. This reveals a weakness of {\dya}: its inability to consistently deliver low-faultline solutions as the population becomes larger. On the other hand, the {\matching} algorithm does not exhibit this weakness, emerging as both the most stable and effective approach.
Finally, as in the previous experiment, the algorithms exhibit a negligible variation over the different samples that we considered
for each value of the parameter.

\begin{figure*}
\centering
\includegraphics[scale=.4]{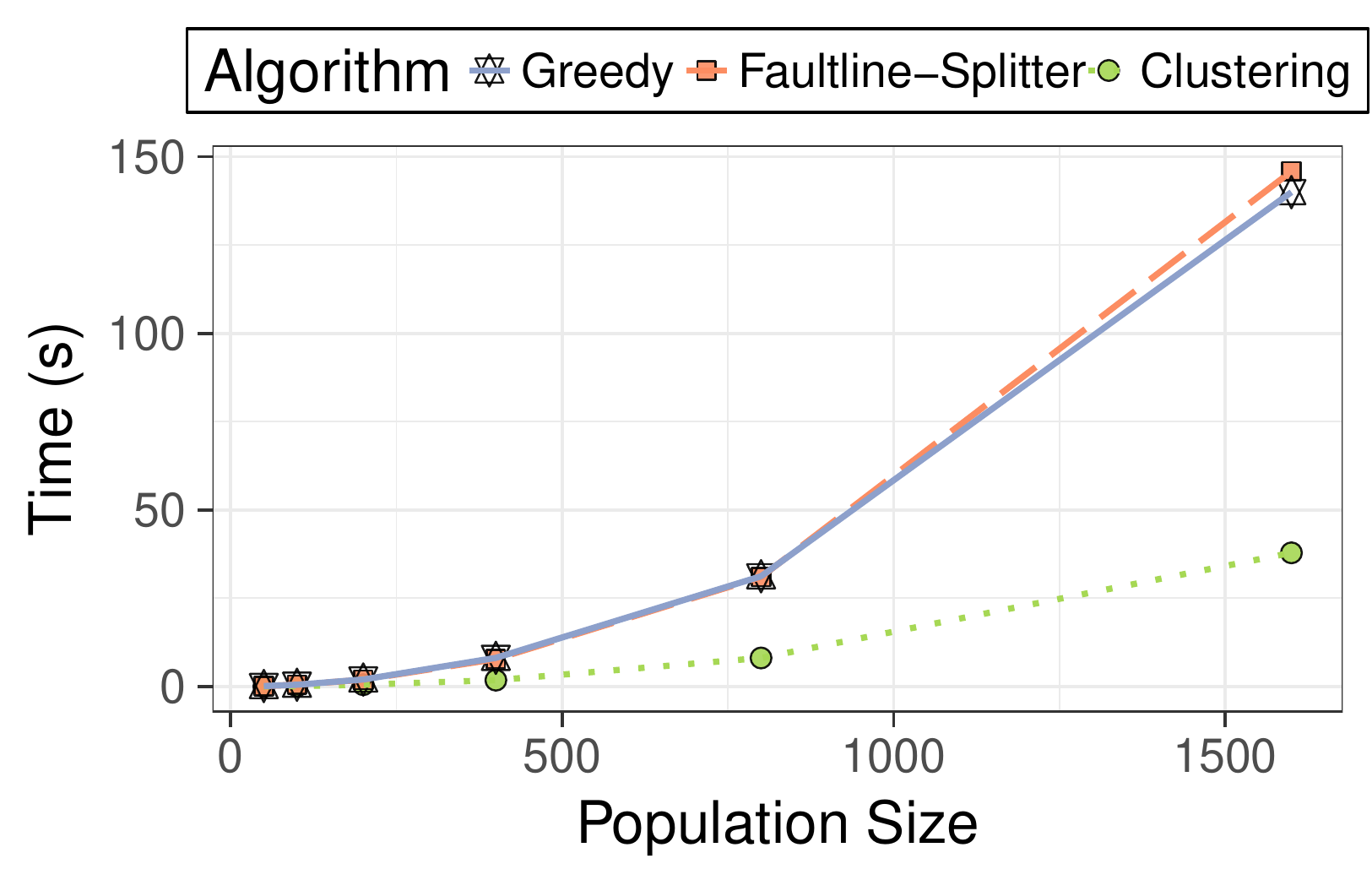}
\caption{\label{fig:runtimes} Running times for different population sizes (parameter $n$)}
\end{figure*}

With respect to computational time, Figure~\ref{fig:runtimes} verifies that {\matching} can scale to large population sizes.
Using the \census\ dataset, we observe that, even for the largest population of $1600$ individuals, the algorithm computed the solution in less 2 minutes.
In fact, its speed was nearly identical to that of the greedy heuristic. Finally, while the {\dya} algorithm emerges as the fastest
option, this comes at the cost of inferior solutions (i.e. teams with higher faultline potential), as we demonstrated in Figure~\ref{fig:population_size}.

\subsubsection{Varying the team size} For this experiment, we set the size of the population of individuals to $|W|=800$.
For each real dataset, we randomly select $100$ populations, of $800$ individuals each, with replacement. We then use the algorithms to
partition each population into teams of size $k$, for $k\in\{3,4,5,\cdots,20\}$. For each algorithm, we report the average normalized
faultline potential achieved over all population for every value of $k$, along with the corresponding $90\%$ confidence intervals. 
The results for \adult, \census\ and \dblp\ datasets are shown in Fig~\ref{fig:team_size}.

\begin{figure*}
\centering
\includegraphics[scale=.4]{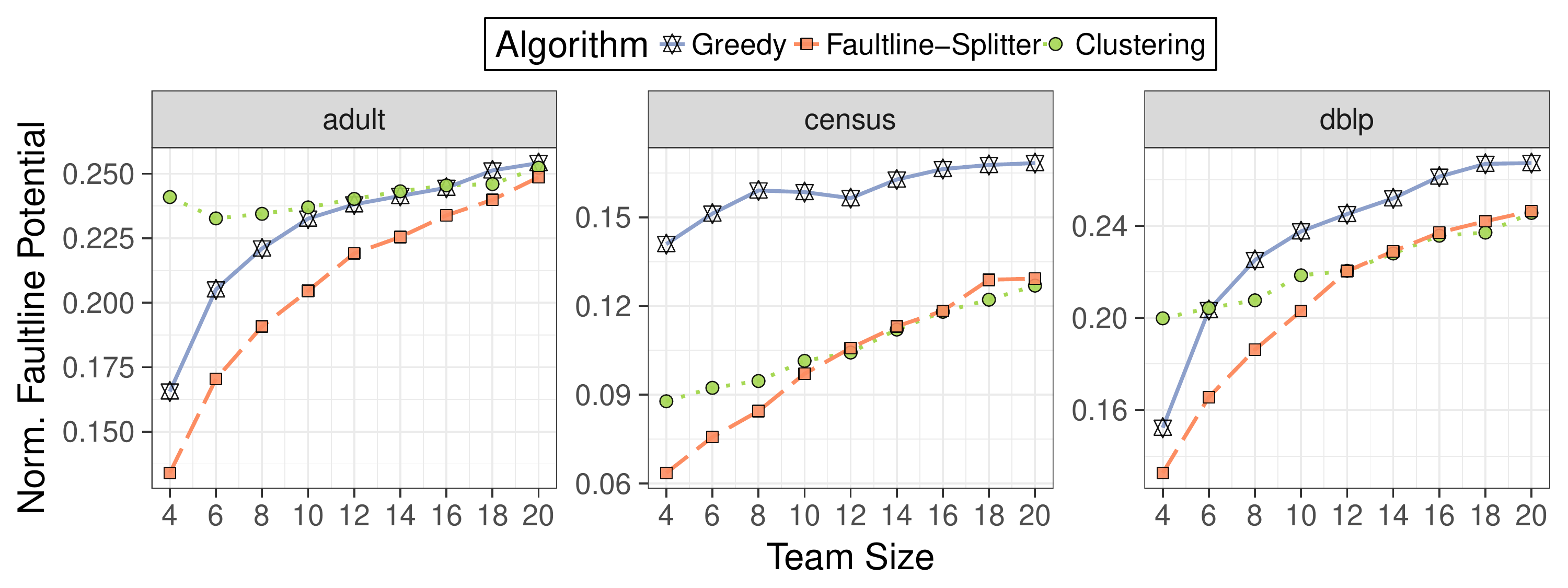}
\caption{\label{fig:team_size} Faultline results for different team sizes (parameter $k$)}
\end{figure*}

We observe that the {\matching} algorithm had the overall best
performance across datasets. We observe that its advantage wanes as the value of $k$ increases. This can be explained by the fact that asking for larger teams makes the problem harder, as it requires the inclusion of additional individuals and thus makes it harder
to avoid the introduction of conflict triangles into the team. This explanation is also consistent with the fact that the performance
of the two algorithms tends to decrease as $k$ becomes larger. A second observation is that the {\greedy} algorithm is consistently outperformed by both {\matching} and {\dya}. This demonstrates the difficulty of the {\partition} problem and the need
for sophisticated partitioning algorithms that go beyond greedy heuristics. Finally, as shown in the figure, we observe that 
the standard deviations for all algorithms were consistently negligible, bolstering our confidence in the reported findings.

\subsubsection{Varying the number of conflict triangles} The purpose of this experiment is to
evaluate the algorithm on populations with different potential for faultlines. 
While random samples obtained from our real-world datasets differ trivially in terms of the 
percentage of conflict triangles, we can engineer synthetic data to obtain
datasets with different number of conflict triangles. To conduct
this experiment, we use the {\synthetic} dataset described in Section~\ref{sec:datasets}.
We consider populations of $400$ individuals and set
the team size equal to $5$. The results are shown in Fig~\ref{fig:bad_freq}.
The plot verifies that finding low-faultline teams becomes harder as the population's inherent
potential for such faultlines increases. However, the {\matching} algorithm consistently outperforms the other
methods. In fact, the gap between the two algorithms increases as the number of conflict triangles
in the population increases. This demonstrates the superiority of the {\matching} algorithm over the other approaches 
in terms of searching the increasingly smaller space of low-faultline solutions.

\begin{figure}[t]
\begin{center}
\includegraphics[scale=.37]{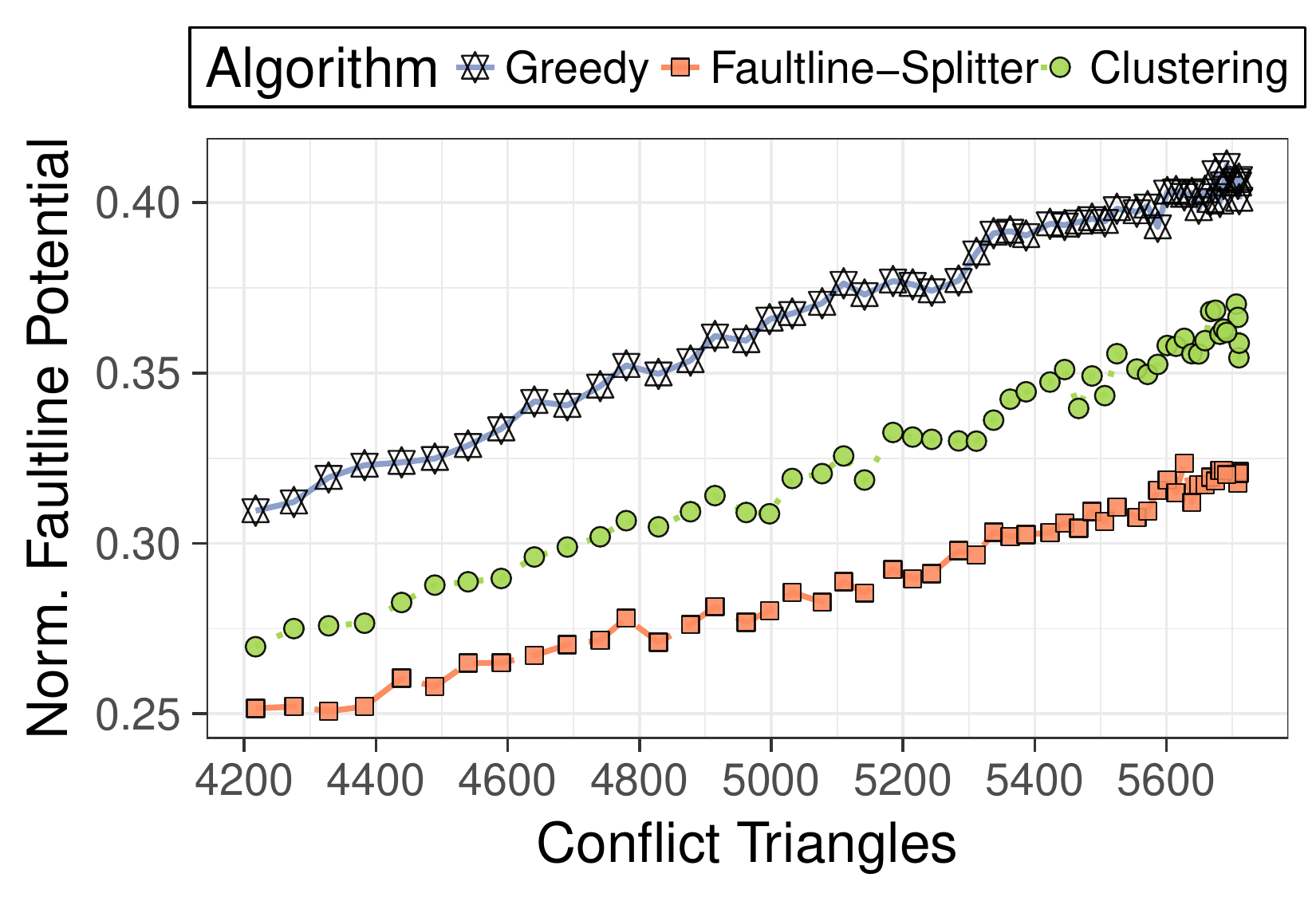}
\end{center}
\caption{\label{fig:bad_freq} Performance of all algorithms for 
synthetic datasets with different number of bad triangles ($400$ workers, teams
of size $5$ and $10$ features)}
\end{figure}

\subsection{Faultline Measurement in existing teams}
\label{exp:meas}
In this section, we compare three alternative options for faultline measurement in existing teams: the proposed {\ct} measure, the {\asw} by~\citep{meyer2013team}, and the Subgroup Strength ({\sss}) measure by~\citep{gibson2003healthy}. We select the {\asw} due to its status as the state-of-the-art, even though, as we discussed in detail in Section~\ref{sec:related}, it is not appropriate for the {\partition}
problem that is the main focus of our work. We select the {\sss} measure because it combines the simplicity and computational efficiency
required for the {\partition} problem with competitive results in previous benchmarks~\citep{meyer2013team}.

\subsubsection{A Comparison on Synthetic Teams}
For this study, we use the {\synthetictwo} dataset which, as we describe in Section~\ref{sec:datasets}, 
includes teams of various sizes and subgroup composition. First, we group the teams according to size. We then use each of the
three faultline measure to evaluate the teams in each group. Finally, we compute the Pearson Correlation Coefficient (PCC) between
every pair of measures. We present the results in Figure~\ref{fig:pearson}. Then, in Figure~\ref{fig:synth_speed} we report the average
computational time needed to compute the score of each team for each of the three measures.

\begin{figure}[ht]
\begin{center}
\begin{tabular}{c}
\subfloat[Pairwise Correlation of Faultline Measures]{
\includegraphics[scale=.35]{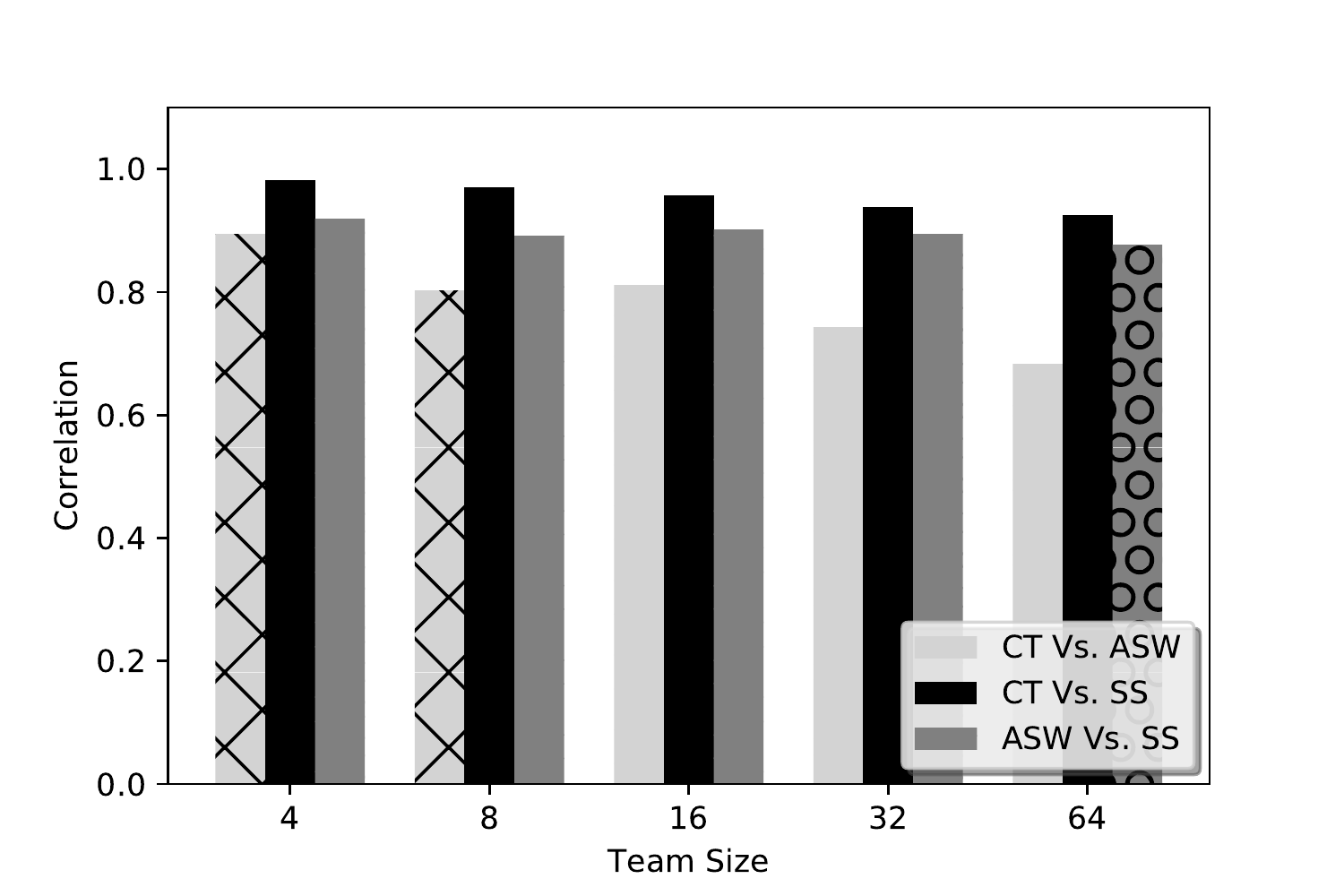}
\label{fig:pearson}
} 
\subfloat[Mean Computational Time]{
\includegraphics[scale=.35]{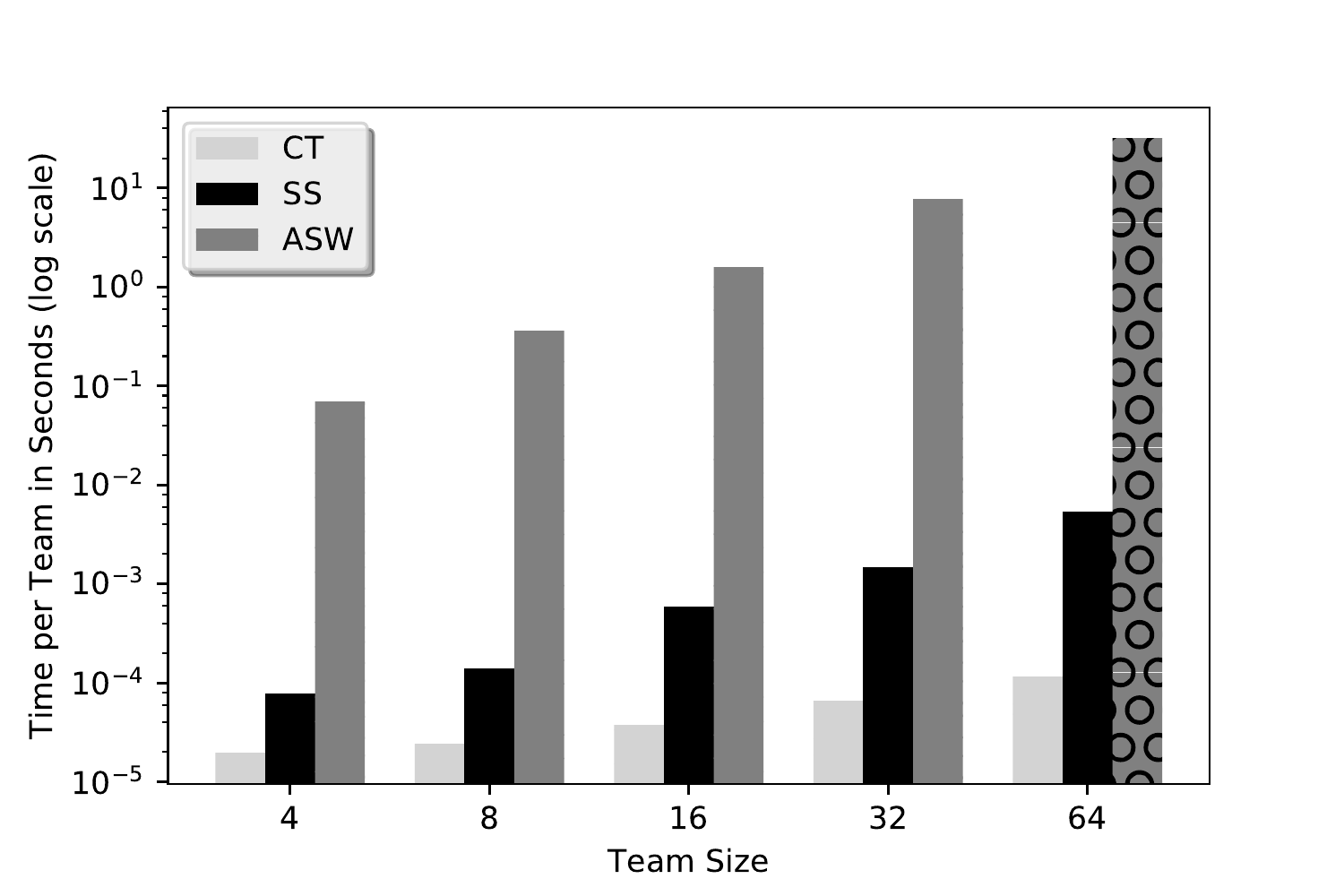}
\label{fig:synth_speed}
}\end{tabular}
\end{center}
\caption{\label{fig:synth}A comparison of the {\asw}, Conflict Triangles ({\ct}), and Subgroup Strength ({\sss}) faultline measures on the 
{\synthetictwo} dataset.}
\end{figure}

The first observation from Figure~\ref{fig:pearson} is that all three measures report similar scores across team sizes, 
with the pairwise PCC over $0.65$. Hence, while the three measures follow different
measurement paradigms, their results tend to be consistent. However, the bars also reveal that the correlation between
{\sss} and {\ct} measures was the highest among all possible measure-pairs. In fact, the observed PCC value for this pair was consistently
around $0.9$, revealing near-perfect correlation. This is intuitive if we consider the nature of the two measures: the conflict triangles counted by the {\ct} measure include, by definition, a pair of team members that are also identified as ``overlapping'' by
the {\sss} measure. A key difference between the two measures is that {\ct} does not consider all-positive triangles (i.e. a triplet of team members with the same value for a feature, see Fig.~\ref{fig:theseauthors-pos}), while {\sss} would consider all 3 dyads in such a triangle as overlaps. However, the results reveal that this difference does not significantly differentiate the results of the two measures, possibly due to the fact that {\sss} does not follow the {\ct}'s counting paradigm and, instead, aggregates overlap sums via the standard deviation.

With respect to computational time, Figure~\ref{fig:synth_speed} verifies the theoretical analysis that we presented in Section~\ref{sec:related}. The y-axis represents the average time (in seconds) required to compute the score for a team, \underline{in log scale}. As we discussed in detail in Section~\ref{sec:problem}, any algorithm for the {\partition} problem has to quickly consider a large number of candidate teams in order to efficiently locate (or approximate) the best possible partitioning. We observe that {\asw} is orders of magnitude slower than the other two measures, with the gap growing rapidly with the size of the teams. In addition, while the {\sss} and {\ct} measures can be easily updated
in constant time as the algorithm makes small changes to the team's roster, this is not the case for {\asw}. In short, while {\asw} may indeed be a competitive option for faultline measurement, our analysis and experiments verify that it is not a good candidate for faultline-optimization problems, such as the one that we study in this work. Out of two fastest measures, 
{\ct} displays a clear advantage over {\sss}. We observe that it is several times faster and, as in the case of {\asw}, the gap grows rapidly
with the size of the population. The results verify the effectiveness of our methodology for computing {\ct}, which we discuss in detail in Section~\ref{sec:eff}. They also demonstrate that, while two measures might satisfy the efficiency principles that are
necessary for efficient faultline-minimization in teams, one of the two can still have a significant computational advantage
that makes it more appropriate for large populations.

\subsubsection{A Comparison on Real Teams}
\label{exp:660}

For this study we use the~\bia\ dataset, which includes two outcomes: (i) the team's performance (represented by its grade) and (ii) the average \emph{satisfaction} of the team's members with their overall collaborative experience. In Figures~\ref{fig:bia1} and ~\ref{fig:bia2} we visualize the performance of each team against its corresponding {\ct}, {\sss}, and {\asw} scores. We observe that performance has a strong
negative association with the {\ct} and {\sss} scores, as demonstrated by the slope of the line. In contrast, the corresponding line for the {\asw}
measure is nearly parallel to the x-axis, suggesting a lack of correlation. This finding is verified by the Pearson Correlation Coefficient (PCC) values for the {\ct}, {\sss}, and {\asw} measures, which were $-0.21$, $-0.23$ and $-0.04$, respectively. Note that a negative correlation
is intuitive, as it means that lower faultlines are associated with higher performance.

In Figures~\ref{fig:bia3} and ~\ref{fig:bia4} we visualize the \emph{satisfaction} of each team against its corresponding {\ct}, {\sss}, and {\asw} scores. The results are consistent with the performance analysis: satisfaction exhibits a strong negative association the {\ct} measure,
while its association with {\asw} is very weak. In fact, the correlation of satisfaction with {\ct} and {\sss} appears to be even stronger than
that of the team's performance. Again, these findings are verified by the PCC values for the {\ct}, {\sss}, and {\asw} measures, which were $-0.26$, $-0.34$ and $0.02$, respectively.

\begin{figure}[ht]
\begin{center}
\begin{tabular}{c}
\subfloat[Performance Vs {\ct}]{
\includegraphics[scale=.35]{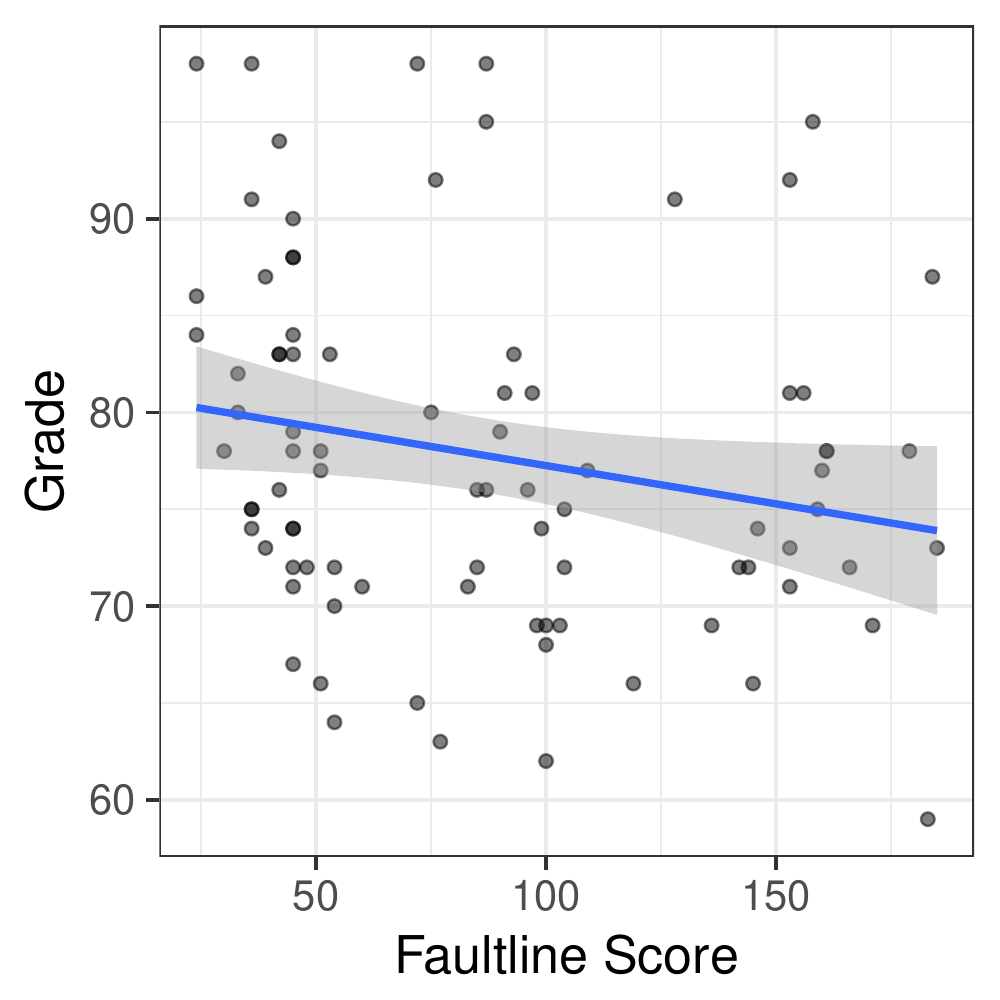}
\label{fig:bia1}
} 
\subfloat[Performance Vs {\asw}]{
\includegraphics[scale=.35]{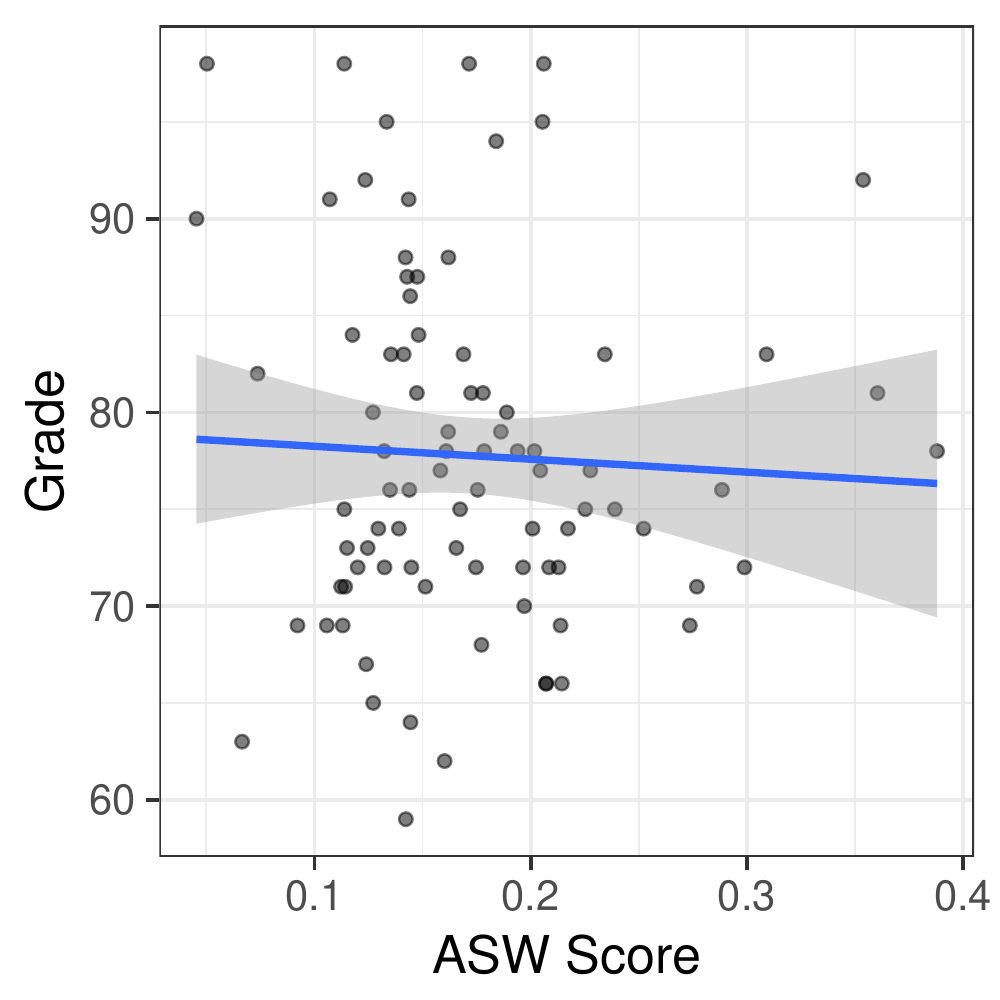}
\label{fig:bia2}
}
\subfloat[Performance Vs {\sss}]{
\includegraphics[scale=.35]{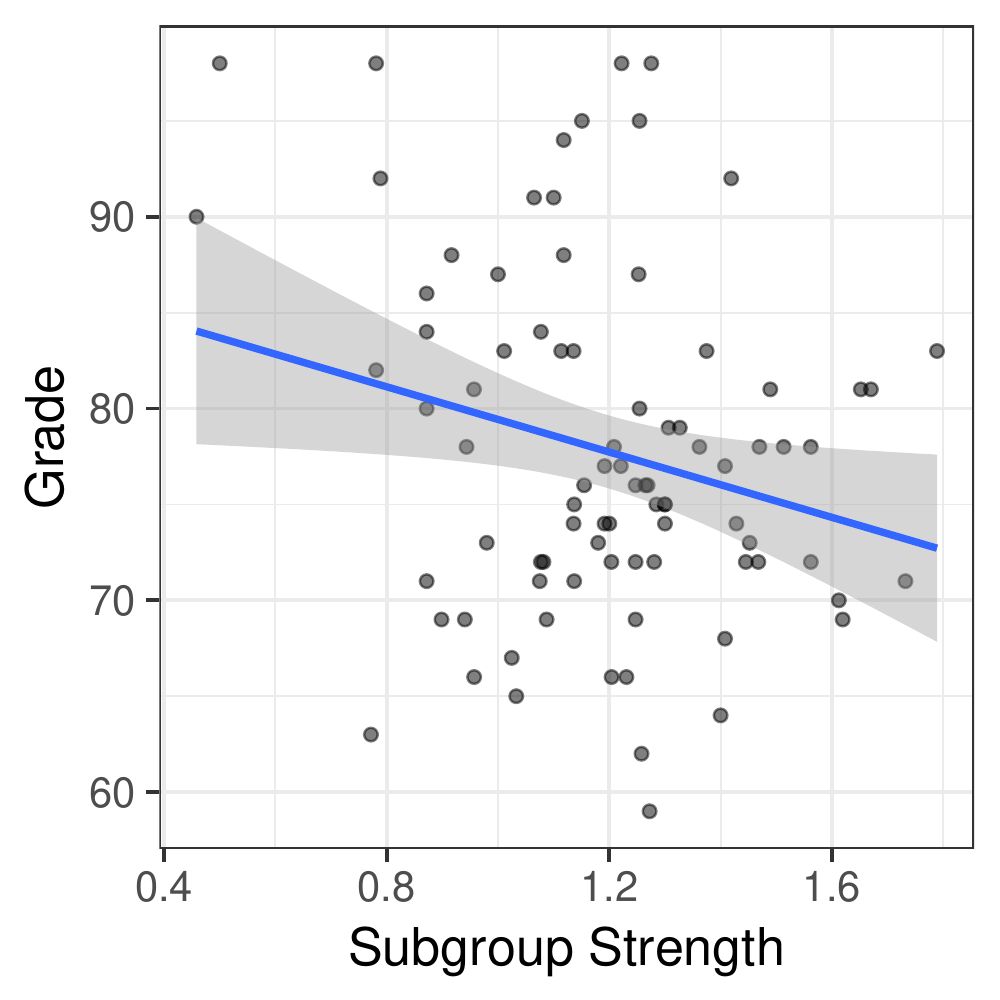}
\label{fig:bia3}
}
\\
\subfloat[Satisfaction Vs {\ct}]{
\includegraphics[scale=.35]{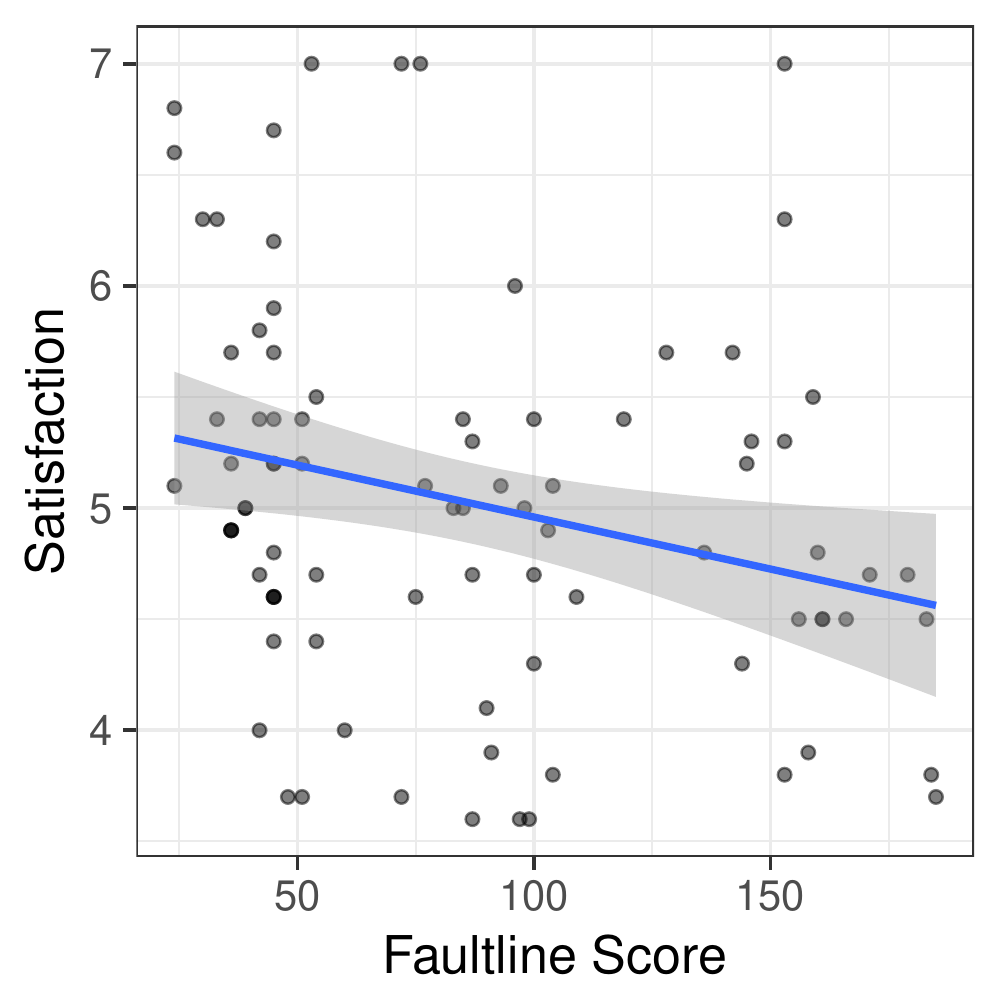}
\label{fig:bia4}
}
\subfloat[Satisfaction Vs {\asw}]{
\includegraphics[scale=.35]{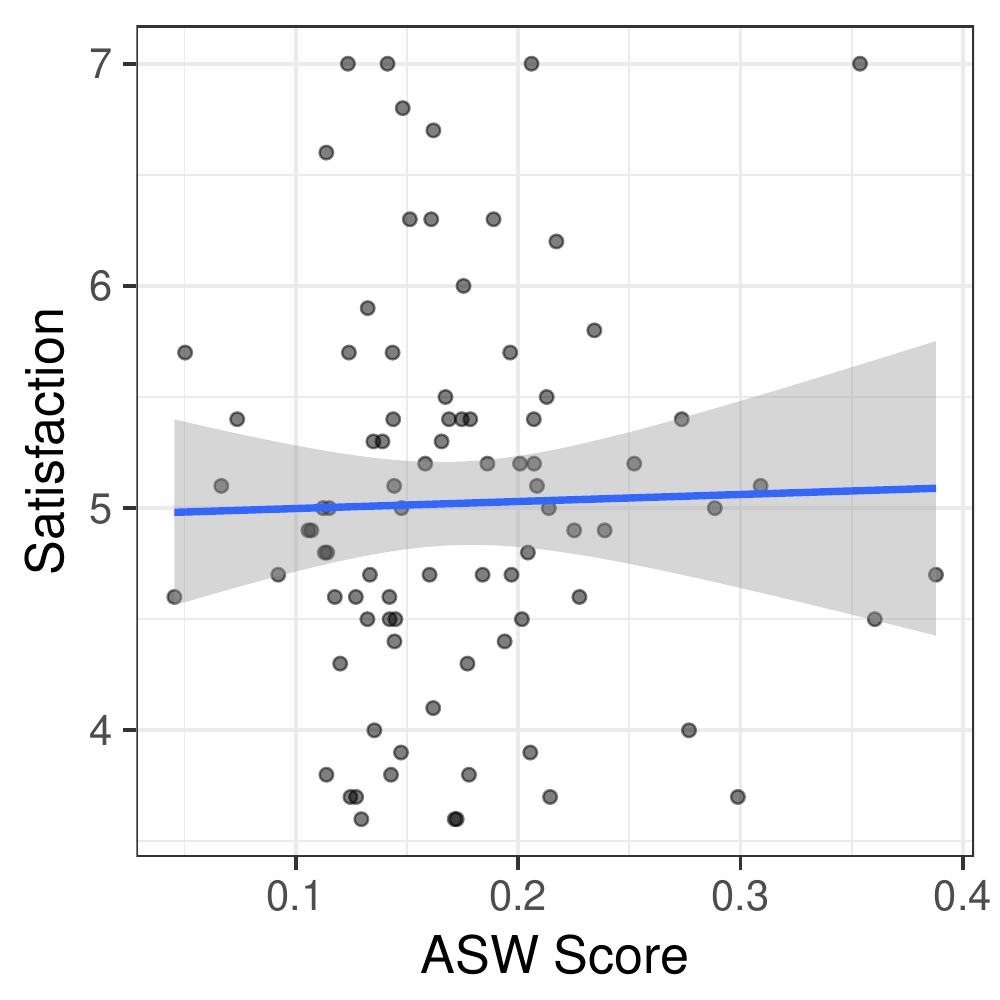}
\label{fig:bia5}
}
\subfloat[Satisfaction Vs {\sss}]{
\includegraphics[scale=.35]{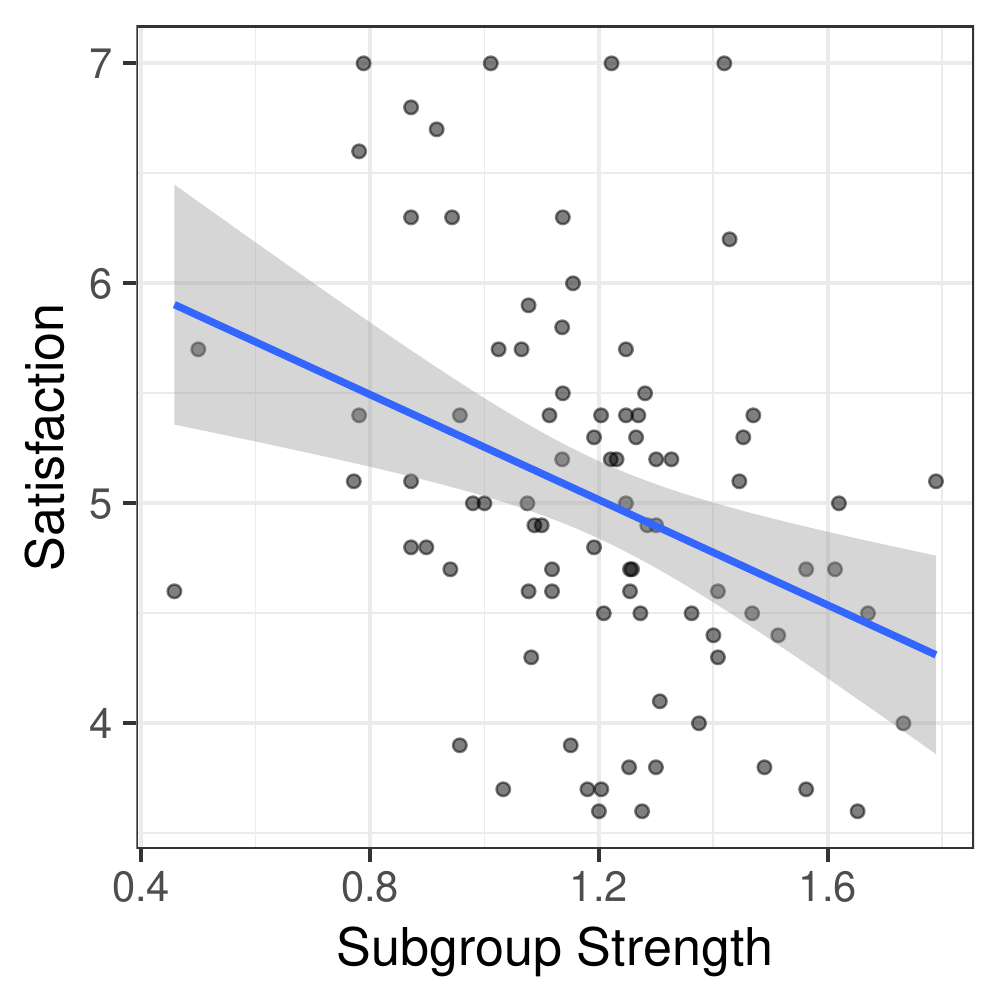}
\label{fig:bia6}}
\end{tabular}
\end{center}
\caption{\label{biall}The association of the {\ct}, {\sss}, and {\asw} measures with team performance and satisfaction.}
\end{figure}

The results verify that the teams' overall faultline-strength, as measured by the {\ct} measure, has a strong
negative association with meaningful outcomes. Next, we demonstrate how a practitioner can examine feature-specific faultlines to identify specific features that are associated with each outcome.

Each of the two outcomes (performance and satisfaction) serves as the dependent variable in a separate regression that also includes the team's faultline potential with respect to different features, according to the {\ct} measure. We also consider multiple
control variables that could account for part of the variance in the dependent variable. We present the results of both regressions in Table~\ref{table:regression}. 

\begin{table}[htp] 
\footnotesize
\centering 
  \caption{Regression Results} 
  \label{} 
   \renewcommand{\arraystretch}{0.65}
   \renewcommand{\arraystretch}{0.9}
\begin{tabular}{@{}l@{}D{.}{.}{-3} @{}D{.}{.}{-3} } 
\\[-1.8ex]\toprule

 & \multicolumn{2}{c}{\textit{Dependent variable:}} \\ 
\cline{2-3} 
\\[-1.8ex] & 
{Grade} & 
{Satisfaction} \\ 
\\[-1.8ex] & 
{(1)} & 
{(2)}\\ 
\hline \\[-1.8ex] 
 Degree & -10.887^{***}  & -1.195^{***} \\ 
  & (-3.793)& (-4.937) \\ 
  & & \\ 
 BS Major & -12.821^{**} & -0.676 \\ 
  & (-2.580 )& (-1.612) \\ 
  & & \\ 
 Gender & 0.106 & 0.147 \\ 
  & (0.028 )& (0.457) \\ 
  & & \\ 
 Country & -8.010^{***} & -1.121^{***} \\ 
  & (-3.001 )& (-4.977) \\ 
  & & \\ 
 ML Exp & -5.428 & -1.029^{*} \\ 
  & (-0.841 )& (-1.889) \\ 
  & & \\ 
 Analytics Exp & 5.233 & -0.063 \\ 
  & (1.082 )& (-0.155) \\ 
  & & \\ 
 Programming Exp & -4.464 & 0.336 \\ 
  & (-0.926 )& (0.827) \\ 
  & & \\ 
 Team Exp & 1.344 & -0.188 \\ 
  & (0.323 )& (-0.535) \\ 
  & & \\ 
 Average ML Exp & 0.219 & -0.105 \\ 
  & (0.185 )& (-1.052) \\ 
  & & \\ 
 Average Analytics Exp & -0.241 & -0.143 \\ 
  & (-0.205 )& (-1.438) \\ 
  & & \\ 
 Average Prog Exp & -1.820^{*}  & 0.014 \\ 
  & (-1.758)& (0.155) \\ 
  & & \\ 
 Average Team Exp & 0.898 & 0.040 \\ 
  & (0.772 )& (0.405) \\ 
  & & \\ 
Constant & 97.821 & 7.308^{***} \\ 
  & (10.586 )& (9.374) \\ 
  & & \\ 
\hline \\[-1.8ex] 
Observations & \multicolumn{1}{c}{86} & \multicolumn{1}{c}{86} \\ 
R$^{2}$ & \multicolumn{1}{c}{0.313} & \multicolumn{1}{c}{0.467} \\ 
Adjusted R$^{2}$ & \multicolumn{1}{c}{0.200} & \multicolumn{1}{c}{0.379} \\ 
Residual Std. Error (df = 73) & \multicolumn{1}{c}{8.126} & \multicolumn{1}{c}{0.686} \\ 
F Statistic (df = 12; 73) & \multicolumn{1}{c}{2.766$^{***}$} & \multicolumn{1}{c}{5.327$^{***}$} \\ 
\bottomrule
\textit{Note:}  & \multicolumn{2}{@{} l}{The dependent variable are grade and} \\ 
\textit{}  & \multicolumn{2}{@{} l}{satisfaction.} \\ 
 \textit{}  & \multicolumn{2}{@{} l}{t-statistics are shown in parentheses.} \\ 
\textit{Significance levels:}  & \multicolumn{2}{@{} l}{$^{*}$p$<$0.1; $^{**}$p$<$0.05; $^{***}$p$<$0.01} \\ 

\end{tabular} 
\label{table:regression}
\end{table}

The table reveals strong negative correlations of the faultline scores for the features \emph{country}, \emph{BS major}, and \emph{current degree} with performance. 
This implies that the existence of potentially conflicting groups in these features can be detrimental to the team's grade. We observe similar trends for the 
\emph{country} and \emph{current degree} features in the context of team satisfaction. Such findings can inform the instructor about the existence of 
potentially problematic dimensions and guide his efforts to strategically design the teams. In practice, this type of regression can be used before solving an instance
of the {\partition} problem, in order to identify the dimensions that need to be considered during the optimization. This is a critical step, as trying to solve
for all possible dimensions is likely to limit the solution space and eliminate high-quality teams due to the existence of faultlines in trivial (non-influential) dimensions.

\section{Generalizing the penalization scheme of aligned conflict triangles}
\label{sec:alignments}
As mentioned earlier our definition of faultline potential (summarized in Equation~\ref{eq:act}) applies, for each conflict triangle, a penalty that is directly proportional to the triangle's alignment across the features. Our experimental results presented in Section~\ref{exp:660} demonstrate that this penalization scheme yields a metric that is a strong predictor of a team's success. However, one might argue that in a specific domain or application, different degrees of alignment should be penalized using a different scheme. In this section, we extend our framework by (1) demonstrating how different penalization scheme of aligned conflict triangles can be implemented, (2) describing a methodology that allows practitioners to \emph{learn} the appropriate penalization scheme for their domain based on information from existing teams in the same domain, (3) studying the \partition\ problem under a given penalization scheme.

\subsection{Faultine potential with a generalized penalization scheme}
Given a team $T$ with a set of features $\mathcal{F}_T$, we define the faultline potential of a team given a penalization 
scheme $g(.)$ as:

\begin{equation}
PCT(T, g)=\sum_{x=1}^{|\mathcal{F}_T|} g(x) \times aligned(x,T),
\label{eq:genact}
\end{equation}

where $aligned(i,T)$ returns the number of conflict triangles that are aligned across exactly $x$ features in $\mathcal{F}_T$. 
The above formulation allows us to flexibly penalize the existence of aligned conflict triangles by selecting the
appropriate $g(x)$ penalty for each value of $x$. Naturally, it makes sense to define $g(x)$ as an ascending function to reflect the fact that higher alignment should translate to a higher faultline potential.
Note that if define $g(x) = x$, then the obtained faultine potential is equivalent to our original definition of $CT(t)$ presented in Equation~\ref{eq:act} (module some constant).  

\subsection{Learning the penalization scheme}
The task of learning the appropriate penalty parameters can be modeled 
as a supervised learning task. Each team serves as a data point in the training set. More specifically, the predictive variables are the $aligned(x,T)$ values for increasing values of $x$. The dependent variable should reflect the degree to which a team's performance is influenced by faultlines. We compute the dependent variables using the following technique. Given a set of teams along with any success metric that encodes their outcome in a particular domain (e.g. performance, satisfaction, cohesion), we obtain the dependent variables by negating the success scores and normalizing them to have a mean equal to $0$ and a standard deviation equal to $1$. The goal is then to learn the penalty-parameters $g(x)$ that best fit the data. To achieve this, we train a linear regression to obtain the best $g(x)$ values. It is important to mention that fitting the linear regression may lead to negative $g(x)$ values. This does not create any issues, but if practitioners desire to obtain faultline potential values that are always positive, they can simply add a constant to all $g(x)$ values. This is a safe operation as it simply adds a constant value to all fautline potential values and does not affect the difference between teams' faultline potentials. In fact, in our experiments we always add a constant value to all $g(x)$ parameters to ensure that $g(0)$ is equal to $0$. This makes the penalization scheme more interpretable as we expect the penalty of conflict-free triangles to be $0$.

If the practitioner has no access to numeric outcomes variables, we can still learn $g(x)$ as follows. The learning task can be modeled as a classification task with a binary variable that is equal $1$ for all actual teams in the data. The training data is then complemented by randomly-populated ``noise'' groups that do not represent actual teams. The binary dependent variable for these fake teams is $0$. In this case, the goal is to find the penalty-parameters that best differentiate between actual and noise teams. This technique builds upon the fact that in most cases, individuals (and managers) tend to form teams that have a lot degree of conflict and faultline potential. 

To demonstrate the effectiveness of our proposed learning procedure, we use the \bia\ dataset as it consists of a set of teams along with two outcome scores, namely ``grade'' and ``satisfaction''. Table~\ref{table:learning} summarizes the $g(x)$ values we obtained using the techniques described above. The first two rows correspond to the $g(x)$ values obtained from the grade and satisfaction metrics. The third row corresponds to $g(x)$ values calculated from our binary classification task (without using any outcome scores). The fourth rows corresponds to $g(x)$ values obtained on a version of \bia\ dataset in which outcome score of each team is randomly sampled from the set $\{0, 1\}$. This row helps verify that the results of the other rows is significant and not due to chance. 

\begin{table*}
  \centering
  \topcaption{Obtained penalization schemes using the \bia\ dataset}
  \begin{tabular}{@{}l|@{}cccccccc@{}}
  \toprule  
   & $\mathbf{g}(1)$ & $\mathbf{g}(2)$ & $\mathbf{g}(3)$ &
                         $\mathbf{g}(4)$ & $\mathbf{g}(5)$ & $\mathbf{g}(6)$ &
                         $\mathbf{g}(7)$ & $\mathbf{g}(8)$\\
  \midrule
  \midrule
  \textbf{Grade} & 0.091 & 0.064 & 0.053 & 0.112 & 0.165 & 0.233 & 0.171 & 0.111 \\
  \textbf{Satisfaction}  & 0.088 & 0.07 & 0.028 & 0.141 & 0.079 & 0.253 & 0.208 & 0.133 \\
  \textbf{Real Vs. Fake}  & 0.068 & 0.099 & 0.079 & 0.061 & 0.115 & 0.184 & 0.223 & 0.171 \\
  \textbf{Random} & -0.063 & 0.021 & 0.428 & -0.041 & -0.153 & 0.053 & 0.142 & 0.098 \\
  \midrule
  \midrule
  \textbf{Frequencies} & 0.2\% & 1.5\% & 6.8\% & 12.9\% & 17.0\% & 11.6\% & 4.5\% & 0.7\%\\
  \bottomrule
  \end{tabular}
  \label{table:learning}
\end{table*}

Note that the first three rows in Table~\ref{table:learning} share a similar trend (and for the most part)
the numbers are ascending representing that the higher degrees of alignment should be penalized more. On the other
hand, we can see that the values in the last row are significantly different and do not exhibit any meaningful pattern.
We can observe that the $g(x)$ values reported in the first $3$ rows, while following the expected trend, sometimes fluctuate. For example, the values of $g(8)$ are smaller than $g(7)$. This can be explained using the last row of the table which summarizes the frequencies of each degree of the alignment in the entire dataset. For instance, we can see that in the entire dataset, there are only 0.7\% of triangles that can form $8$ aligned conflicts. This means, that in our learning task this value is in almost all cases set to $0$ for both successful and unsuccessful teams. Thus, the parameters learned using the linear regression are more subject to noise. In fact, if we focus only on degrees of alignments that have at least 5\% presence in the data, we can see that the $g(x)$ values are more robust and conform to our expected behaviour. 

\subsection{Team-formation under the generalized penalization scheme}
As we discussed in Section~\ref{sec:operationalize}, solving the \partition\ problem for a large
group of individuals requires an operationalized notion of faultline that can be (1) computed in
linear time and (2) updated in constant time when a member joins or leaves the team. Unfortunately, these
two criteria may not hold for a given penalization scheme. In fact given a team $T$, computing the $PCT(T, g)$
requires a running time of $O(m|T|^3)$. This is because our speed-up technique described 
in Section~\ref{sec:eff} can not be applied to any penalization scheme. This makes the \partition\ problem
even more challenging to solve as it becomes computationally expensive. The \matching\ algorithm can still 
be used to solve the \partition\ problem given any penalization scheme, but the solution does not scale
up to large population of individuals. Given that, we present some theoretical and experimental evidence to
demonstrate that solving the \partition\ problem with our original penalization scheme produces teams that are
of high-quality under different penalization schemes as well. Of course, directly solving the \partition\ problem with
a given penalization scheme can produce better results, but in most cases the slight improvement can not justify
the huge required computational cost.

Let us use $CT(T)$ and $PCT(T, g)$ to refer to the definition of faultline potential (according to Equation~\ref{eq:act}) and the faultline potential given a penalization scheme
(according to Equation~\ref{eq:genact}) respectively. Now, it is easy to show that $$CT(T) * \max(g(x)) * m \ge PCT(T, g).$$
The above equation simply states that in the worst-case scenario all $m$ features of conflicting individuals form a conflicting triangle. This is an strict upper bound for $PCT(T, g)$.
Although this may not be a tight bound, it suggest that optimizing $CT(T)$ directly might be an efficient strategy for solving the \partition\ problem under any penalization scheme.

The following experiment further demonstrates that optimizing the original faultline potential (presented in Equation~\ref{eq:act} is quite aligned with optimizing
faultline potential under a given penalization scheme. In this experiment, we have solved the \partition\ problem on the \bia\ dataset using the penalization scheme from the first row of Table~\ref{table:learning}. More precisely, we ran the \matching\ algorithm to create $50$ teams of equal size. In each iteration of the algorithm, we recorded the faultline potential according to Equation~\ref{eq:act}. Figure~\ref{fig:ctcompare} illustrates how the value of $CT(T)$ and $PCT(T, g)$ compare as the optimization proceeds. We can see that the $PCT(T, g)$ has an almost linear relationship with our original definition of faultline potential. This
implies that by solving the \partition\ problem using our original penalization scheme, we can benefit from the speed-up techniques we introduced in Section~\ref{sec:eff} without sacrificing the quality of
the obtained teams even if a different penalization scheme is desired. 

\begin{figure}
    \centering
    \includegraphics[scale=0.45]{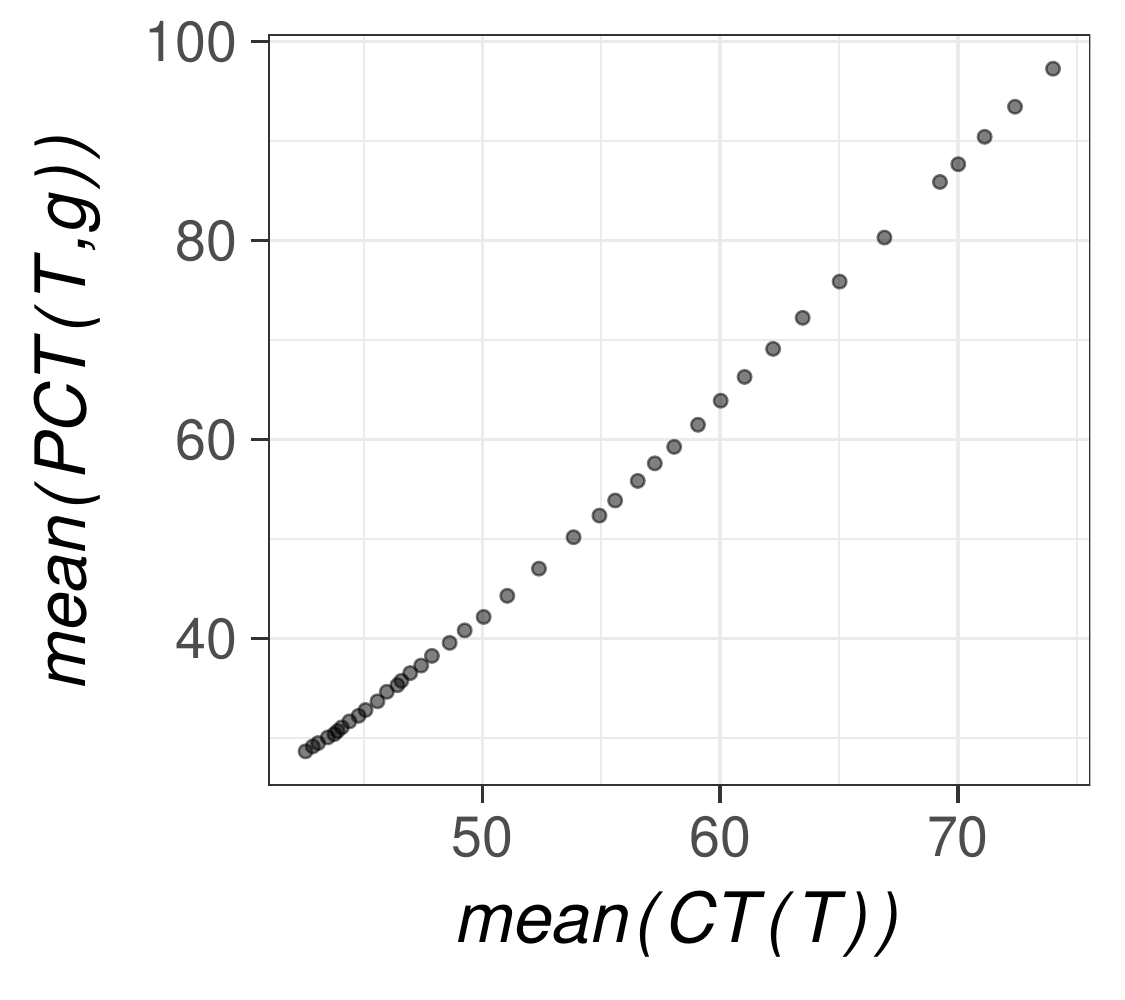}
    \caption{Comparing the faultline potential using different penalization schemes}
    \label{fig:ctcompare}
\end{figure}


\section{Handling numeric attributes}
\label{sec:numer}
One of the limitations of {\ct} is that it is
primarily designed for nominal attributes. Thus, numerical attributes need to be discretized into bins
prior to computing the faultline score. The ability to handle multimodal data is a well-known challenge in faultline
measurement. For instance, the popular {\asw} approach has to pre-process the data by using dummy
variables to encode categorical variables as numeric. Next, we present two techniques to extend our
basic {\ct} model to deal with numerical attributes.

The first technique is based on binning, but aims to creates bins of variable
length that can accurately capture the distribution of the underlying data. More precisely, a pre-processing
module based on Kernel Density Estimation (KDE) could automate the discretization process and deliver
dynamic segmentations that accurately capture the distribution of numeric variables~\citep{rudemo1982empirical}.
The resulting bins would then represent the natural groups of numeric values that are present in the given dataset.

An alternative technique that departs from the standard binning paradigm 
would be to use a threshold $\gamma$ to define agreement and disagreement between team members. 
Specifically, given a numeric feature $f$, we say two
individuals $i$ and $i'$ are in agreement \emph{iff} $|\worker_i(f) - \worker_{i'}(f)| \le \gamma$. Otherwise, the 
two individuals are considered to be in a disagreement. As before, a triangle is identified as a
conflict triangle with respect to feature $f$ if two of each members agree on feature $f$ and disagree with the third individual in the triangle. The problem then translates into the task of selecting an appropriate value for $\gamma$. Domain knowledge is a key factor
in this effort, as each feature is likely to have its own threshold. For instance, while a difference of 2 years for the \emph{age}
feature is generally considered small, a difference of 2 stars in the context of the popular 5-star rating scale is far more significant. An intuitive way to set feature-specific thresholds would be to assume that two members agree on feature $f$ 
if the difference of their corresponding values is within 1 standard deviation of the same feature (as computed across the entire population that we want to partition into teams). A second way to tune the feature-specific thresholds is to use a validation set that includes the scores of teams for meaningful team outcomes, such as performance or satisfaction. We used such a dataset in Section~\ref{exp:660}. We can then choose the threshold values that maximize the correlation between the resulting faultline and outcomes scores.

While the above methods allow us to flexibly model (dis)agreements and address numeric attributes during the computation of conflict triangles, they do not directly model the \emph{degree} of disagreement between two team members in the context of a numeric feature.
For instance, a conflict triangle with two members in their 20s and one in their 30s tends to be less problematic than a triangle with two member in their 20s and one in their 60s.
To address this issue, we can weigh (the disagreement in) a conflict triangle by directly using the 
numeric values of its members. In practice, the weight of a conflict triangle with respect to feature $f$
would then be equal to the the average absolute difference between the values of feature $f$ for the two individuals in disagreement. The {\ct} measure would
then be expressed as a weighted sum, rather than the pure number of conflict triangles in the given team. 
Combining this method
with the two techniques that we discussed above (or with any techniques based on binning or definitions of disagreement) 
enables us to comprehensively extend our approach to handle numeric attributes.

\section{Discussion}\label{sec:conclusions}
Our work focuses on the previously unexplored overlap between the decades of work on team faultlines and the rapidly
growing literature on automated team-formation. We formally define the {\partition} problem, which is the first problem definition that asks for the formation of teams with minimized faultlines from a large population of candidates. We present a detailed complexity analysis and introduce a new faultline-minimization algorithm ({\matching}) that outperforms competitive baselines in an experimental 
evaluation on both real and synthetic data.  

One of the major challenges that we address in this work is finding a faultline measure
that can be efficiently applied to faultline optimization. As we highlight in this paper, computational efficiency (in a practical team-formation setting) translates into two requirements that an appropriate measure should satisfy: (i) the ability to compute the faultline score of a team in linear time, and (ii) the ability to update a team's score in constant time after small changes to the team (e.g. the removal or addition of a member).
The relevant literature has described multiple operationalizations of the faultline concept. However, as we discuss in detail in Section~\ref{sec:related}, these operationalizations do not satisfy these requirements and are only appropriate for measuring faultline strength in \emph{existing} teams. As such, they are not scalable enough to serve as the objective function of a combinatorial algorithm that has to process a large population and evaluate very large numbers of candidate-teams in order to find a faultline-minimizing solution. Therefore, we introduce a new measure that we refer to as \emph{Conflict Triangles} ({\ct}). The {\ct} measure is based on the extensive literature on modeling social structures and is consistent with the fundamental principles of faultline theory by~\citep{lau1998demographic}. In addition, {\ct} satisfies the two efficiency requirements and is appropriate for 
faultline-optimization algorithms.

\subsection{Implications}
Our work is the first to incorporate the faultline concept into an algorithmic framework for automated team-formation. 
From a team-builder's perspective, the ability to control the faultlines of teams that are automatically sampled from 
a large population of candidates has multiple uses. First, it allows the team builder to proactively reduce the risk 
of undesirable outcomes that have been consistently linked with faultlines, such as conflicts, polarization, and disintegration. Second, it provides an effective way to manage the diversity within a team. A trivial way to eliminate faultlines
is to create highly homogeneous teams. However, this approach would also lead to teams that are unable to benefit
from the well-documented benefits of diversity, such as innovation and increased performance~\citep{kearney2009and,roberge2010recognizing,van2003joint}. In order to avoid such shortcomings, a team-builder can utilize our algorithmic framework 
to strategically engineer low-faultline teams without over-penalizing diversity. A characteristic example is a team that
is maximally diverse; a team in which no two individuals share a common attribute. Consistent with the faultline theory
by~\citep{lau1998demographic}, our framework would recognize this as a team with the same faultline potential as a 
perfectly homogeneous team. We demonstrate this via examples in Figures~\ref{fig:theseauthors} and~\ref{fig:groups}.

Our team-partitioning paradigm has applications in both an organizational and educational setting. In a firm setting, the task
of partitioning a workforce into teams is common. By using the proposed {\matching} algorithm, a manager can identify 
faultline-minimizing partitionings within the multidimensional space defined by various employee features. A regression analysis, such 
as the one we described in Section~\ref{exp:660}, can guide the manager's team-building efforts by selecting specific features with potentially problematic faultlines. In a classroom setting, instructors often face the task of partitioning their students into teams
for assignments and projects. As we demonstrated in our experiments, faultlines in student teams can have a strong association with meaningful outcomes, such as performance and member satisfaction. By releasing our team-partitioning software, we hope that we can
automate this team-formation task and benefit both students and instructors.

\subsection{Directions for Future Work}
Future work could focus on algorithms
that combine faultlines minimization (either as an objective function or via constraints) with other
factors, such as intra-team communication, skill coverage, and recruitment cost. Such work would add
to the rapidly growing literature on automated team formation, which we review in Section~\ref{sec:auto}.
We expect this to be a challenging task from an optimization perspective, as additional constraints
can be hard to satisfy while trying to avoid the creation of faultlines. 
For instance, if the distribution of skills is strongly correlated with the population's demographics,
a homogenous team is unlikely to exhibit a diverse skillset. Hence, the ability to leverage both homogeneity
and diversity will be an asset for such efforts. 

The proposed {\matching} algorithm can be combined with any faultline measure that follows the
efficiency principles that we describe in this work (i.e. linear computation and constant
updates). Future work on such measures is essential, as existing measures are not scalable
enough for optimization purposes. We make our own contribution in this direction via The
{\ct} measure that we propose in this work.

In conclusion, we hope that future efforts will be able to build on our work to address challenging problems
that combine efficient algorithmic constructs for automated team-formation with the rich findings on the causes and effects of teams faultlines. 

\section*{Acknowledgement}
This research was supported in part by NSF grants IIS-1813406 and CAREER-1253393.

\bibliographystyle{abbrvnat}

\end{document}